\documentclass[journal]{IEEEtran}
\usepackage{xcolor}
\usepackage{amsfonts}
\usepackage[dvips]{graphicx}
\usepackage{epsfig,latexsym}
\usepackage{float}
\usepackage{indentfirst}
\usepackage{amssymb}
\usepackage{amsmath}
\usepackage{times}
\usepackage{subfigure}
\usepackage{psfrag}
\usepackage{cite}
\usepackage{cases}
\usepackage{amsmath}
\usepackage{subfigure}
\usepackage{cite}
\usepackage{setspace}
\usepackage{algorithm}
\usepackage{algorithmic}
\usepackage{multirow}
\usepackage{amsmath}
\usepackage{xcolor}
\usepackage{mathrsfs}
\usepackage{dsfont}
\usepackage{bm}
\usepackage{multicol}
\usepackage{amsfonts,amssymb}
\usepackage{array,color}
\usepackage{bbm}
\usepackage{algorithmic,algorithm}
\newtheorem{theorem}{$\mathbf{Theorem}$}
\newtheorem{lemma}{$\mathbf{Lemma}$}
\newtheorem{proposition}{Proposition}

\newtheorem{Definition}{Definition}

\newtheorem{remark}{$\mathbf{Remark}$}

\begin{document}

\markboth{Submit to IEEE Trans. Multimedia, vol. XX, no. Y, Month 2015}
{Peng: H-CRAN \ldots}

\title{Energy-Efficient Resource Allocation Optimization for Multimedia Heterogeneous Cloud Radio Access
Networks} \vspace{-5pt}
\author{\IEEEauthorblockN{Mugen~Peng,~\IEEEmembership{Senior Member,~IEEE}, Yuling~Yu, Hongyu~Xiang, and H. Vincent Poor,~\IEEEmembership{Fellow,~IEEE}}
\thanks{Mugen~Peng, Yuling~Yu, and Hongyu~Xiang are with the Key Laboratory of Universal Wireless Communications for Ministry of Education, Beijing University of Posts and Telecommunications, Beijing, China (e-mail: {\tt pmg@bupt.edu.cn, aliceyu1215@gmail.com, xianghongyu88@gmail.com}).}
\thanks{H. V. Poor is with the School of Engineering and Applied Science, Princeton
University, Princeton, NJ, USA (e-mail: {\tt poor@princeton.edu}).}}

\vspace{-5pt}

\maketitle

\vspace{-5pt}

\begin{abstract}
The heterogeneous cloud radio access network (H-CRAN) is a promising
paradigm which incorporates the cloud computing into heterogeneous
networks (HetNets), thereby taking full advantage of cloud radio
access networks (C-RANs) and HetNets. Characterizing the cooperative
beamforming with fronthaul capacity and queue stability constraints
is critical for multimedia applications to improving energy efficiency (EE) in H-CRANs. An
energy-efficient optimization objective function with individual
fronthaul capacity and inter-tier interference constraints is
presented in this paper for queue-aware multimedia H-CRANs. To solve this
non-convex objective function, a stochastic optimization problem
is reformulated by introducing the general Lyapunov optimization
framework. Under the Lyapunov framework, this optimization problem
is equivalent to an optimal network-wide cooperative beamformer
design algorithm with instantaneous power, average power and
inter-tier interference constraints, which can be regarded as the
weighted sum EE maximization problem and solved by a generalized
weighted minimum mean square error approach. The mathematical
analysis and simulation results demonstrate that a tradeoff between
EE and queuing delay can be achieved, and this tradeoff strictly
depends on the fronthaul constraint.
\end{abstract}


\section{Introduction}

With the explosive growth of mobile multimedia traffic demand and number of mobile devices,
the next-generation wireless networks face significant challenges in improving system capacity and
guaranteeing users' quality of service (QoS). Cloud radio access networks (C-RANs) present a promising approach to these
challenges by curtailing both capital and operating expenditures for providing mobile multimedia applications, while providing high
energy-efficiency and capacity \cite{CRAN1}\cite{CRAN}. In C-RANs, the traditional base station (BS)
is decoupled into the distributed remote radio heads (RRHs) and the baseband unit (BBU).
Antennas are equipped with RRHs to transmit/receive radio frequency (RF) signals,
and BBUs are clustered as a BBU pool in a centralized location with aggregating all BS computational
resources, which provides large-scale processing and management functions for the signals transmitted/received from RRHs.
With this architecture, mobile operators can easily expand and upgrade the network by deploying additional
RRHs, and thus the corresponding operational costs can be greatly reduced.

The heterogeneous cloud radio access network (H-CRAN) is regarded as a new paradigm to meet
performance requirements of the fifth generation (5G) cellular system for mobile multimedia applications by incorporating cloud computing
into heterogeneous networks (HetNets)\textcolor[rgb]{1.00,0.00,0.00}{\cite{HCRAN1}\cite{HCRAN2}},
in which the control and user planes are decoupled. The existing macro base station (MBS) that has been deployed in traditional cellular networks is used to alleviate capacity constraints over the fronthaul and provide seamless coverage with QoS guarantees for users. {\color{red}In particular, burst multimedia traffic and real-time multimedia traffic with low-bit transmit rate can be efficiently served by the MBS.} For control signaling and system
data broadcasting at MBSs, it alleviates the capacity and time delay
constraints in the fronthaul links between RRHs and the BBU pool, and allows
RRHs to use sleep mode efficiently to decrease energy consumption. {\color{red}RRHs are preferred to provide both real-time and non-real time multimedia applications with high speed data rates, such as real-time interactive high quality video, delay-tolerant web browsing, non-real time video or massive file download, etc.} With the help of MBSs, RRHs can be used to provide only the high-capacity service and are transparent to the served users. Note that the radio signal processing for all RRHs is executed in the BBU pool, while for the MBS is implemented locally. The inter-tier interference between the BBU pool and the MBS can be mitigated by the distributed coordinated multi-point (CoMP) transmission and reception technique.
Comparing with C-RANs and HetNets, H-CRANs have been demonstrated to significant performance
gains though advanced collaborative signal processing and cooperative radio resource allocation are still
challenging\textcolor[rgb]{1.00,0.00,0.00}{\cite{HCRAN1}}.

Intuitively, cloud computing in the BBU pool based on
large-scale cooperative signal processing can suppress intra-tier
interference and achieve significant cooperative gains in H-CRANs.
{\color{red}
The inter-tier interference to RRHs from the MBS equipped
with multiple antennas can be suppressed by coordinated
scheduling or cooperative multiple-input multiple-output
(MIMO) techniques, which substantially improves the spectral efficiency (SE). For instance, inter-tier interference
can be suppressed by using zero-forcing, which results from the
aggressive spatial multiplexing \cite{peng_IEEEACCESS}.} Such characteristics in
H-CRANs bring challenges to optimize the overall SE or energy
efficiency (EE) because too many factors and challenges must be
jointly considered, such as collaborative signal processing to
suppress intra-tier and inter-tier interference in the physical
(PHY) layer, and cooperative radio resource allocation and queue-aware
packet scheduling in the medium access control (MAC) and upper layers.
In addition, capacity constraints of fronthaul and backhaul links
must be considered as well.

\vspace{-10pt}
\subsection{Related Work}


Much attention has been paid to resource allocation in C-RANs
recently. In\textcolor[rgb]{1.00,0.00,0.00}{\cite{I:A1}}, to
minimize the network power consumption, a greedy RRH on$/$off
selection algorithm has been proposed to maximize the reduction in
the network power consumption at each step.
In\textcolor[rgb]{1.00,0.00,0.00}{\cite{I:A2}}, an antenna selection
scheme that jointly optimizes the antenna selection, regularization
factor and power allocation has been presented to maximize the
averaged weighted sum-rate in large-scale C-RAN downlink systems.
The joint optimization of MIMO and discontinuous transmission (DTX)
with practical implementation constraints has been investigated
in\textcolor[rgb]{1.00,0.00,0.00}{\cite{I:A3}} to improve EE
performance. Meanwhile,\textcolor[rgb]{1.00,0.00,0.00}{\cite{I:A4}} has proposed a joint
cell association and beamformer design algorithm for downlink and
uplink C-RANs. Clearly, these characteristics and achievements to
improve SE and EE performance of C-RANs should be further enhanced
in H-CRANs, in which the cell association with RRH/MBS and the
inter-tier interference should be additionally considered.

Meanwhile, a number of studies have considered the SE and EE optimization
of HetNets, in which radio resource allocation with the
consideration of inter-tier interference is often the primary focus.
In\textcolor[rgb]{1.00,0.00,0.00}{\cite{I:A5}}, an EE optimization
problem with statistical quality of service (QoS) constraints in
orthogonal frequency-division multiple access (OFDMA) systems
has been analyzed, and a subchannel grouping scheme to obtain a
closed form solution has been presented, which is simplified to a
multi-target single-channel optimization problem by using the
channel-matrix singular value decomposition method.
In\textcolor[rgb]{1.00,0.00,0.00}{\cite{I:A6}}, to improve EE in
heterogeneous cognitive femtocell networks, a spectrum sharing and
resource allocation scheme has been formulated as a Stackelberg
game, and a gradient based iterative algorithm has been proposed to
achieve the Stackelberg equilibrium solution.
In\textcolor[rgb]{1.00,0.00,0.00}{\cite{I:A7}}, an energy-efficient
partial spectrum reuse (PSR) scheme has been proposed. Since the
optimal PSR factor, defined as the portion of spectrum reused by
micro cells in two-tier heterogeneous networks, is not in an
explicit form generally, a closed-form limit of the optimal PSR
factor has been derived as the ratio of the user rate requirement
to the entire system spectrum bandwidth. Numerical results showed
that adopting PSR can reduce the network energy consumption by up to
50\% when the transmit power of MBSs is 10dB higher than that of low
power nodes (LPNs). An energy efficient precoding for
coordinated multi-point transmission under constraints of individual
date rate requirements from each user, maximal transmit power of each
base station (BS), and zero-forcing (ZF) has been obtained by
introducing the subspace decomposition method
in\textcolor[rgb]{1.00,0.00,0.00}{\cite{I:A8}}, where the
performance gain of the proposed ZF precoder has been verified by
comparing with several existing optimal linear precoders.

The aforementioned works typically assume that all users are
time-delay insensitive and neglect special QoS requirements for
delay sensitive users, which may suffer serious performance
deterioration caused by the large service delay.
{\color{red}
To minimize the user's queuing time and
achieve a degree of fairness, the authors of{\color{red}~\cite{LZhou}} designed an efficient blind scheduling policy that performs well across magnitudes of fairness, simplicity,
and asymptotic optimality for a relatively general mobile media cloud.
}
In\textcolor[rgb]{1.00,0.00,0.00}{\cite{I:A9}}, a joint power and
rate control algorithm with average delay constraints has been
proposed and solved by a game-theoretic model.
To study energy efficiency-delay tradeoffs in multiple-access networks,
a game-theoretic approach is proposed in\textcolor[rgb]{1.00,0.00,0.00}{\cite{I:V_Poor}},
where each user seeks to choose a transmit power that maximizes its own utility
while satisfying its delay requirements.
To deal with the
co-channel interference problem and the individual statistical delay
QoS guarantee problem, the cross-layer optimization of a two-tier
underlay HetNet has been studied
in\textcolor[rgb]{1.00,0.00,0.00}{\cite{V:taomeixia}} using
large deviation theory, in which the cross-layer optimization
problem is transformed into a long term weighted sum effective
capacity maximization problem. In addition, to further guarantee the
service fairness between different delay-tolerant users, the queue
backlogs maintained at the transmitters for each user have to be
considered in the design of radio resource optimization schemes. The
authors of\textcolor[rgb]{1.00,0.00,0.00}{\cite{I:A10}} have
investigated the delay-optimal policy in a two-user multiple access
channel, where the delay minimization problem is formulated as a
Markov decision process (MDP) and the optimal policy traded a
portion of the sum-rate for balancing the queue lengths to minimize
the average delay. The delay-optimal power and subcarrier allocation
problem for OFDMA systems was modeled as a $K$-dimensional infinite
horizon average reward MDP with the control actions based on channel
state information (CSI) and joint queue state information (QSI)
in\textcolor[rgb]{1.00,0.00,0.00}{\cite{V:cuiyinglau}}. Furthermore,
a cache-enabled cross-layer opportunistic cooperative MIMO framework
for wireless video streaming is proposed
in\textcolor[rgb]{1.00,0.00,0.00}{\cite{{V:anliu}}}. By equipping
the relay with a cache to buffer the video streams, the cache
control policy adaptive to the popularity of the video files could
provide more cooperative opportunities, while the power control
policy adaptive to QSI and CSI is determined by solving the
approximated MDP approach using the continuous time Bellman
equations to maintain the QoS metrics of playback interruption
probability and the buffer overflow probability.

However, the number of queues in realistic systems is not often
sufficiently large, which causes the issue of curse of
dimensionality with the MDP approach due to the exponential growth
of the cardinality of the system state space. In addition, it is
difficult to obtain a distributed resource allocation solution with
MDP since the potential function is not decomposable. To achieve a
desired tradeoff between network throughput and queuing
delay,\textcolor[rgb]{1.00,0.00,0.00}{\cite{I:A11}}
and\textcolor[rgb]{1.00,0.00,0.00}{\cite{I:A12}} proposed
distributed resource allocation and user scheduling solutions in
Long Term Evolution-Advanced (LTE-A) relay networks and wireless
multihop networks, respectively, both of which utilize Lyapunov
optimization to stabilize the queues of networks when
optimizing performance metrics. Lyapunov optimization is a useful tool for handling
queue-aware radio resource allocation problems with a
good balance between performance and implementation complexity. For the Lyapunov optimization theory,
the concepts of Lyapunov function and Lyapunov drift are introduced, and the performance metrics of the
network can be optimized while stabilizing queues of the network by
greedily minimizing the drift-plus-penalty. With
the Lyapunov optimization approach, the problem of opportunistic
cooperation in a cognitive two-tier underlay HetNet has been studied
in\textcolor[rgb]{1.00,0.00,0.00}{\cite{V:mikeneely}}, where the
cognitive LPNs handle intelligent access admission, cooperation
decision making, and power control to maximize their own throughputs subject
to average power constraints. The obtained online control algorithm
can stabilize the multimedia traffic queue without requiring any knowledge of
the multimedia traffic arrival rates. A two-stage queue-aware cross-layer radio
resource allocation algorithm has been proposed
in\textcolor[rgb]{1.00,0.00,0.00}{\cite{V:juntingchen}} based on
minimizing drift-plus-utility, which can be easily applied for
HetNets.

Inspired by{\color{red}\cite{I:A11}--\cite{V:juntingchen}}, the
well-developed stability theory of Lyapunov optimization
is considered in this paper. Since the optimal radio resource
allocation policy can be achieved by minimizing the
drift-plus-utility, in which the resulting queuing delay and utility
performance are bounded, we focus on maximizing EE under the stable
queue with the transmit power, individual fronthaul capacity, and
interference constraints in H-CRANs. There are three technical
challenges associated with this EE optimization problem in
queue-aware H-CRANs:

\begin{itemize}
 \item \textbf{Challenges due to Joint Considerations of EE and Queuing Delay:} Unlike other
 works focusing only on optimizing SE or EE, which only requires
 CSI in the PHY layer, the optimization involving EE
 and average queuing delay is fundamentally challenging since it
 introduces coupling between the PHY and MAC layers. To take the queuing delay into consideration, the resource
 allocation policy should be a function of both QSI and CSI, which is
 a nontrivial problem since the QSI and CSI vary randomly at each time slot and may not have a closed-form
 expression.
 \item \textbf{Challenges due to Inter-tier Interference and Individual Fronthaul Capacity Constraints:}
 Unlike traditional C-RANs, the inter-tier interference from
 MBSs in H-CRANs should be suppressed by advanced collaborative processing techniques, and the inter-tier interference to MUEs
 should be mitigated to a low level with advanced coordinated
 scheduling and power control techniques. Unlike the traditional HetNets, intra-tier interference
 among RRHs in H-CRANs can be suppressed through the centralized BBU pool but with individual
 fronthaul capacity constraints, and the inter-tier interference often exists between a single powerful MBS and a very large number of RRHs.
 \item \textbf{Challenges due to Minimization of Drift-plus-penalty:}
 The queues of data flows are coupled due to the mutual
 inter-tier interference in H-CRANs. Thus, the associated stochastic optimization problem formulated
 by minimizing the drift-plus-penalty under the Lyapunov optimization framework
 is complex because resource allocation decisions for different RRHs are affected by
 each other. With the time-varying nature of wireless environments,
 this problem is challenging to solve.
\end{itemize}

\vspace{-10pt}
\subsection{Contribution and Organization}

With the introduction of H-CRANs, the development of effective radio resource management techniques
to optimize EE for non-real time packet service is important. In addition, to satisfy
diverse QoS requirements, it is crucial to use cross-layer radio resource management
algorithms in H-CRANs, which has been seldom studied to date.
In this paper, a weighted EE
performance metric is presented, {\color{red}and the corresponding EE
optimization problem with inter-tier interference,
individual fronthaul capacity, and total transmission power constraints is
formulated, in which both cross-layer design and queue-aware congestion control are taken into account. Since this EE optimization problem is a combination of
time-averaged variables and instantaneous variables, an optimal network-wide cooperative beamformer design
algorithm is proposed based on minimizing the drift-plus-utility under the
Lyapunov optimization framework, in which a generalized weighted
minimum mean square error (WMMSE){\color{red}~\cite{WMMSE}} approach
is used to optimize the EE performance and guarantee the queue
stability.} Although the WMMSE approach has been applied
in{\color{red}~\cite{Yuwei}} to jointly optimize the user scheduling
and beamforming vectors under either dynamic or fixed BS clustering
for C-RANs, this previous work does not explicitly take the EE
optimization and queue stability into consideration for H-CRANs.

The major contributions of this paper can be summarized as follows.
\begin{itemize}
\item An average weighted EE utility function in terms of sum transmit rate and total energy consumption with different weight
factors is defined to capture the EE performance in H-CRANs. To
maximize this EE performance metric and keep the multimedia traffic
queue stability, an EE optimization problem with the instantaneous
and average power, individual fronthaul capacity and inter-tier
interference constraints is formulated for queue-aware H-CRANs. To
solve this non-convex optimization problem, the Lyapunov
optimization framework is utilized, under which the optimization
problem is transformed into the minimization of the
drift-plus-penalty function. Furthermore, this minimization of the
drift-plus-penalty function can be reformulated as the optimal
network-wide beamformer design problem under transmit power and
inter-tier interference constraints.

\item A generalized WMMSE approach is proposed to solve the optimal network-wide beamformer design
problem. Unlike previous work in{\color{red}~\cite{Yuwei}},
whose aim is to solve the weighted sum rate maximization problem
with backhaul constraints in C-RANs, this paper applies the
generalized WMMSE approach to solve the average weighted EE utility
objective function with each RRH's transmit power, individual
fronthaul capacity, and inter-tier interference constraints. To
quantitatively optimize the tradeoff between the average weighted EE
and the queuing delay on demand, a non-negative parameter $V$ is
defined, which in turn adaptively affects the solutions of
network-wide cooperative beamformer design and power allocation.

\item The proposed optimal network-wide beamformer design algorithm can approach an $[\mathcal{O}(1/V),\mathcal{O}(V)]$ tradeoff between the averaged weighted EE performance and queue backlog, which indicates that the
average weighted EE performance can be arbitrarily close to the
optimum with the gap of $\mathcal{O}(1/V)$ at the expense of
incurring an average queue backlog that is $\mathcal{O}(V)$.
Simulation results exhibit the EE performance under varied
$V$, and show the influence of the fronthaul capacity constraint on
the tradeoff between the average weighted EE performance and average
queue backlog.

\end{itemize}

The rest of this paper is organized as follows. Section II describes
the system model and formulates the optimization problem. In Section
III, based on the general Lyapunov optimization framework, the
formulated problem is transformed into a WMMSE problem which can be
solved by an iterative approach. In Section IV, the performance of
the proposed network-wide beamformer design algorithm is analyzed.
Section V presents the simulation results and Section VI summarizes
this paper.

Throughout this paper, the following notation is adopted. {\color{red}
Lower-case bold letters ${\bf{v}}$ denote column vectors, and
upper-case bold letters ${\bf{D}}$ denote matrices.} We use
$\mathbb{C}$ to denote the complex domain. The complex Gaussian
distribution is represented by $\mathcal{CN(\cdot,\cdot)}$, while
$\bf{Re\{\cdot\}}$ stands for the real part of a scalar. We use
$\{x\}^{+}$ to denote the larger of $x$ and $0$.
$\mathbb{E}\big[x\big]$ is the expectation of the
random variable $x$, and $(\cdot)^H$ denotes the matrix conjugate
transpose. $\mathbf{0}_{N}$ and $\mathbf{I}_N$ are $N \times N$
zero matrix and identity matrix, respectively.

\section{System Model and Problem Formulation}

In this section, the considered H-CRAN system scenario and
definition of network stability are introduced. Based on the
H-CRAN system model and defined queue stability, the EE
optimization problem is formulated.

\vspace{-10pt}
\subsection{System Model}

As illustrated in Fig. \ref{Sys}, a downlink H-CRAN system
consisting of one MBS and $N$ RRHs is considered. $N$ RRHs are
deployed within the same coverage of the single MBS in an underlay
manner. The RRHs and MBS are connected to a BBU pool with the
fronthaul and backhaul links, respectively. The MBS and each RRH are
equipped with ${L_M}$ and ${L_R}$ antennas, respectively. Define the
set of MBS and RRHs as $\left\{ {0,1,2,...,N} \right\}$, where the index
$0$ refers to the MBS, which serves $K_M$ single-antenna MBS user
equipments (MUEs), and ${\mathcal{N}}=\left\{ {1,2,...,N} \right\}$
denotes the set of RRHs, which cooperatively serve $K_R$
single-antenna RRH user equipments (RUEs) with user-centric
clustering. Define the set of RUEs as $\mathcal{K}_R =\left\{
{1,2,...,K_R} \right\}$, and the set of MUEs as
$\mathcal{K}_M=\left\{ {1,2,...,K_M} \right\}$. This H-CRAN system
is assumed to operate in the slotted time mode with the unit time
slot $t \in \left\{ {0, 1, 2, \cdots } \right\}$, where the time
slot $t$ refers to the interval $\left[ {t,t + 1} \right)$. Under
the assumption that the BBU pool centrally processes all RUEs'
signals and distributes each RUE's data to an individually selected
cluster of RRHs through fronthaul links, each RUE is cooperatively
served by its serving cluster through the joint beamforming
technique, and receives an independent data stream from the RRH at
the time slot $t$. It is assumed that the scalar-valued data stream
${s_k}(t)$ is temporally white with zero mean and unit variance.

\begin{figure}
\centering
\includegraphics[width=9cm]{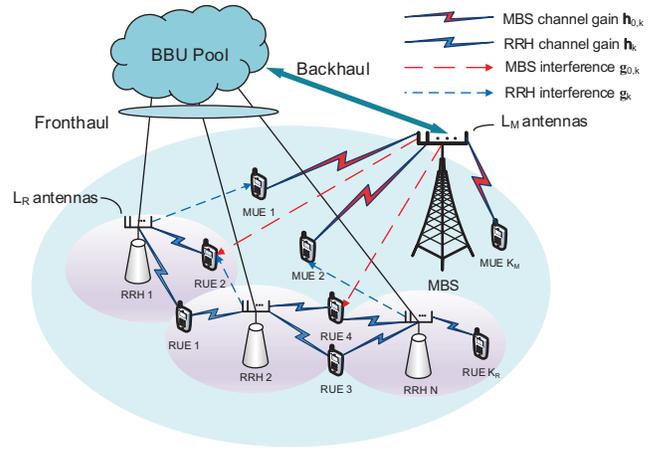}
\vspace*{-10pt} \caption{Downlink H-CRANs with one MBS, $N$ RRHs,
$K_R$ RUEs and $K_M$ MUEs.}\vspace*{-10pt} \label{Sys}
\end{figure}

Under centralized large-scale cooperative processing in the BBU
pool, the transmit beamformer from RRH $n$ to RUE $k$ in the time
slot $t$ is defined as ${{\bf{v}}_{n,k}(t)} \in \mathbb{C}^{{L_R}
\times 1}$, and the corresponding network-wide beamforming vector
for RUE $k$ can be expressed as ${{\bf{v}}_k}(t) =
[\mathbf{v}_{1,k}^H(t),\mathbf{v}_{2,k}^H(t),\ldots,\mathbf{v}_{N,k}^H(t)]^H\in
\mathbb{C}^{N{L_R} \times 1}$. Given that ${{\mathbf{D}}_n} =
\left\{
{\underbrace{\mathbf{0}_{L_R},\ldots,\mathbf{0}_{L_R},}_{n-1}{\mathbf{I}_{{L_R}}},\ldots,\mathbf{0}_{L_R}}
\right\} \in \mathbb{C}{^{{L_R} \times N{L_R}}} (n > 0)$,
${{\bf{v}}_{n,k}(t)}$ can be represented through ${{\bf{v}}_k(t)}$, i.e.,

\begin{equation}
\label{eq:I0} {{\bf{v}}_{n,k}(t)} = {{\mathbf{D}}_n}{{\bf{v}}_k(t)},
\end{equation}

Note that ${{\bf{v}}_k}(t)$ can be written as a combination of the transmit power ($\|{{\bf{v}}_k}(t)\|^2$)
and the unit beamformer (${{\bf{\bar{v}}}_k}(t)=\frac{{{\bf{v}}_k}(t)}{\|{{\bf{v}}_k}(t)\|}$) for simplicity.

In particular, the transmit beamformer from the MBS to MUE $k$ is
denoted by ${{\bf{v}}_{0,k}(t)} \in \mathbb{C}^{{L_M} \times 1}$.
Though all RRHs can potentially serve each scheduled RUE, in fact,
each RUE is mainly contributed to by only a small number of adjacent
RRHs and the network-wide beamforming vector is often group
sparse{\color{red}~\cite{Yuwei}}.

With the linear transmit beamforming scheme at the
RRHs{\color{red}~\cite{linearBF}}, {\color{red}the received signal at the RUE
$k$, denoted by $y_k(t) \in \mathbb{C}$, consists of the desired signal, the interference signal of other RUEs, and the interference signal of the total $K_M$ MUEs. As a result, $y_k(t) \in \mathbb{C}$ can be written as}

\begin{equation}
\label{eq:I1}
\begin{split}
{y_k}(t) = & {{\bf{h}}_k^H(t)}{{\bf{v}}_k(t)}{s_k(t)} + \sum\limits_{j
= 1,j \ne k}^{{K_R}} {{{\bf{h}}_k^H(t)}{{\bf{v}}_j(t)}{s_j(t)}} \\
&+\sum\limits_{i = 1}^{{K_M}}
{{{\bf{g}}_{0,k}^H(t)}{{\bf{v}}_{0,i}(t)}{s_i(t)}} + {n_k}(t),
\end{split}
\end{equation}
{\color{red}where ${{\bf{h}}_k(t)} \in \mathbb{C}^{N{L_R}\times 1}$ denotes the
CSI matrix from all RRHs' transmit antennas to the RUE $k$, and
${{\bf{g}}_{0,k}(t)} \in \mathbb{C}^{{L_M}\times 1}$ denotes the
CSI matrix from the MBS's transmit antennas to the RUE $k$.}
${n_k}(t)$ is the received noise at the RUE $k$ with the
distribution $\mathcal{CN}\left( {0,{\sigma ^2}} \right)$, where
$\sigma^2$ is the noise variance at RUEs. {\color{red}Eq. (\ref{eq:I1}) suggests
that both ${\bf{v}}_j(t)$ and ${\bf{v}}_{0,i}(t)$ should be carefully
designed to suppress the intra-tier and inter-tier interference,
respectively.}

\vspace{-10pt}
\subsection{Queueing Model}

Since the MBS in H-CRANs is mainly used to deliver the control signalling
and provide seamless coverage with a low bit rate, for the ease of implementation, the beamformers of the MBS can be assumed to be
fixed over a longer duration than the scheduling slot of RRHs, and
thus the performance of MUEs can be assumed to remain stable if the inter-tier interference from
RRHs is suppressed to a pre-defined threshold. Therefore, we can focus only on the performance optimization
with queue stability for overall RRHs in H-CRANs under the condition that the queue stability of the MBS is guaranteed.

Suppose there are queues maintained for RUEs in H-CRANs which are
represented by ${\bf{Q}}\left( t \right) = \left\{ \left. {Q_k (t)}
\right| k = 1,...,K_R \right\}$, where ${Q_k \left( t \right)}$
denotes the queue backlog for RUE $k$ at time slot $t$. The
random multimedia traffic arrival for RUE $k$ at the time slot $t$ is
denoted by $A_k(t)$, which is assumed to be independent and
identically distributed (i.i.d.) over time slots with the peak
arrival rate $A_{k}^{max}$. Define the set of $A_k(t)$ is as
${\bf{A}}\left( t \right) = \left\{ \left. {A_k (t)} \right| k =
1,...,K_R \right\}$, and the arrival rates of queues are $
\boldsymbol{\lambda}={\mathbb{E}}\left\{ {{\mathop{\rm
{\bf{A}}}\nolimits} \left( t \right)} \right\} $.

At each time slot, the arrival and departure rates of RUE $k$
are $ A_k \left( t \right)$ and $ R_k \left( t \right)$,
respectively. Therefore, $ Q_k \left( t \right)$ evolves according
to

\begin{equation}
\label{eq:I3}
{Q_k}(t + 1) = {\left\{ {{Q_k}(t) - {R_k}(t)}
\right\}^ + } + {A_k}(t).
\end{equation}

Considering the random and bursty characteristics of multimedia traffic
arrivals and the QoS requirement of RUEs in H-CRANs, it is
imperative to consider queue-aware resource allocation
techniques. Therefore, to achieve this objective, the queue
stability is defined as follows.

{\color{red}\begin{Definition}
\textit{A discrete time process $Q(t)$ is mean-rate stable{\color{red}~{\cite{II:R1}}} if}
\begin{equation}
\label{eq:stability}
\mathop {\lim}\limits_{t \to \infty}{\frac{\mathbb
{E}\{|Q(t)|\}}{t}}=0.
\end{equation}
\end{Definition}
Note that an absolute value of
$Q(t)$ is used in the mean rate stability definition, which is useful for virtual queues those can be possibly negative.
}
\vspace{-10pt}
\subsection{Problem Formulation}

To optimize the EE performance of H-CRANs, transmission rate and
power consumption performance metrics should be jointly considered.
According to (\ref{eq:I1}), these two performance metrics are
both presented as functions of the network-wide beamforming vector
${\bf{v}}_k(t)$ for RUE $k$:
\begin{itemize}
\item[\emph{$\bullet$}] \emph{Transmission Rate:} {\color{red}RUE $k$ is scheduled at time slot $t$, i.e., $R_k(t)$ is
nonzero if and only if its network-wide beamforming vector
${\bf{v}}_k(t)$ is nonzero.} {\color{red}Based on the network-wide beamforming vector ${\bf{v}}_k(t)$ for RUE $k$ and ${\bf{v}}_j(t)$ for RUE $j$, the signal-to-interference-plus-noise ratio (SINR) can be directly derived as $\mathbf{v}_{k}^{H}(t)
\mathbf{h}_{k}(t)\big(\!\sum_{j=1,j\neq
k}^{K_{R}}\!\mathbf{h}_{k}^{H}(t)\mathbf{v}_{j}(t)
\mathbf{v}_{j}^{H}(t)\mathbf{h}_{k}(t)+ \phi_k(t)
\big)^{-1}\mathbf{h}_{k}^H(t)\mathbf{v}_{k}(t)$. As a result, according to the Shannon capacity formula,} the achievable transmission rate for
RUE $k$ at time slot $t$ can be expressed as

\begin{equation}
\label{eq:I2}
\begin{split}
R_{k}(t)=& \log_{2}\bigg(1+\mathbf{v}_{k}^{H}(t)
\mathbf{h}_{k}(t)\big(\!\sum_{j=1,j\neq
k}^{K_{R}}\!\mathbf{h}_{k}^{H}(t)\mathbf{v}_{j}(t)\\
&\mathbf{v}_{j}^{H}(t)\mathbf{h}_{k}(t)+ \phi_k(t)
\big)^{-1}\mathbf{h}_{k}^{H}(t)\mathbf{v}_{k}(t)\bigg),
\end{split}
\end{equation}
where $\phi_k(t)=
\sum_{i=1}^{K_{M}}\mathbf{g}_{0,k}^{H}(t)\mathbf{v}_{0,i}(t)\mathbf{v}_{0,i}^H(t)\mathbf{g}_{0,k}(t)
+\sigma^{2}$ can be assumed to remain constant in time slot $t$
because the beamformers of the MBS are fixed.

\item[\emph{$\bullet$}] \emph{Power Consumption:} Since the transmitted signals from RRHs to RUEs have unit variance,
the radio frequency power consumption ${P_n} (t)$ of the $n$-th RRH
in the time slot $t$ depends only on the beamformer transmitting to the
RUEs. Therefore, the power consumption for RRH $n$ serving all
potential $K_R$ RUEs can be written as
\begin{equation}
\label{eq:I5}
{P_n} (t) = \sum\limits_{k = 1}^{{K_R}} {{\bf{v}}_k^H(t) {{{\mathbf{D}}_n^H}{{\mathbf{D}}_n} } {{\bf{v}}_k(t)}}+PC_n(t)+PF_n(t),
\end{equation}
where $P_n(t)$ denotes the total power consumption of the $n$-th RRH,
and $PC_n(t)$ and $PF_n(t)$ are the circuit power consumption and fronthaul power consumption
of RRH $n$, respectively. Note that the circuit
power of RRHs is negligible because the energy consumption of
air conditioning is avoided. Since RRHs are connected to the BBU pool via optical
fiber to alleviate fronthaul capacity constraints, the power consumption of the fronthaul is rather small compared with the transmit power of RRHs, and it can be neglected, too. Thus, the power consumption
model can be reformulated as
\begin{equation}
\label{eq:I5}
\begin{split}
{P_n} (t) & = \sum\limits_{k = 1}^{{K_R}} {{\bf{v}}_k^H(t) {{{\mathbf{D}}_n^H}{{\mathbf{D}}_n} } {{\bf{v}}_k(t)}}, \\
{\overline P _n} & = \mathop {\lim }\limits_{t \to \infty }
\frac{1}{t}\sum\limits_{\tau = 0}^{t - 1} \mathbb{E}\{{P_n}(\tau
)\},
\end{split}
\end{equation}
where ${\overline P _n}$ is the time average of $P_n(t)$.

\end{itemize}

{\color{red}The traditional EE metric is defined as the ratio of the weighted
sum transmit rate to the corresponding weighted total
energy consumption in units of bit/Hz/Joule, which is given by

\begin{equation}
\label{EE_T}
\begin{split}
\tilde{\eta}_{{EE}} (t)=\frac{\sum\limits_{k =
1}^{K_R}  {\omega'_kR_k}( t)}
{\sum\limits_{n = 1}^N \mu' _nP_n(t)},
\end{split}
\end{equation}
where ${\omega'_k} \ge 0 $ and ${\mu' _n} \ge 0 $ represent the
transmit weight of user $k$ and the power consumption weight of the
$n$-th RRH, respectively.
}

Following \cite{II:R2} and\cite{II:R3}, instead of directly maximizing $\tilde{\eta}_{{EE}} \left( t \right)$,
we define an alternative form of EE, ${\eta}_{EE}\left( t \right)$, and aim at maximizing it.

\begin{Definition}
\textit{To quantitatively capture the
relative importance of transmit rate and power consumption, the
weighted EE utility function $f\left(R_k(t),P_n(t)\right)$ in terms
of sum transmit rate and total energy consumption with varied weight
factors is used \cite{II:R2} in this paper to denote the equivalent EE metric of
overall RRHs in the time slot $t$ as follows:}
\begin{equation}
\label{eq:I6}
\begin{split}
{\eta}_{EE}\left( t \right) = & f\left(R_k(t),P_n(t)\right) \\
=&
\frac{\alpha }{K_R}\sum\limits_{k = 1}^{K_R} {{\omega_k}{{R_k}\left( t
\right)}} - \frac{{1 - \alpha }}{N}\sum\limits_{n = 1}^N {{\mu
_n}{{P_n}\left( t \right)}},
\end{split}
\end{equation}
\textit{where $\alpha \in \left[ {0,1} \right]$ is a weighting factor
representing the ratio of the transmit rate to the power
consumption. Here, ${\omega_k} \ge 0 (channels/bits)$ and ${\mu _n} \ge 0 (W^{-1})$ represent the
transmit weight of user $k$ and the power consumption weight of the
$n$-th RRH, respectively.}
\end{Definition}

\begin{remark}
(1) Note that ${\eta}_{EE}\left( t \right)$ can be used to characterize
$\tilde{\eta}_{{EE}} \left( t \right) $ [\cite{II:R2}, cf. from Eq. (9) to Eq. (11) on Page 3].
(2) Since both $R_k(t)$ and $P_n(t)$
depend on the network-wide beamforming vector ${\bf{v}}_k(t)$,
${\eta}_{EE}\left( t \right)$ is mainly determined by
${\bf{v}}_k(t)$, and the optimization of ${\eta}_{EE}\left( t
\right)$ is strictly related to ${\bf{v}}_k(t)$.
\end{remark}

Note that the beamformers from RRHs to RUEs cause severe
inter-tier interference to MUEs. Therefore, the inter-tier
interference from RRHs to the $k$-th MUE should be constrained,
which can be formulated as
\begin{align}
&\sum\limits_{j = 1}^{{K_R}}
{{\bf{v}}_j^H(t){\bf{g}}_k(t){{\bf{g}}_k^H(t)}{{\bf{v}}_j(t)}} \le
{\varphi _k},k \in {\mathcal{K}_M}, \forall t,
\label{eq:interference}
\end{align}
{\color{red}where ${\bf{g}}_k (t)\in \mathbb{C}^{N{L_R}\times 1}$ denotes the CSI
matrix from all RRHs' transmit antennas to MUE $k$.} {\color{red}Eq. (\ref{eq:interference}) suggests that the overall inter-tier
interference from adjacent RRHs should be suppressed to a
pre-defined threshold by appropriately designing the network-wide
beamforming vector ${\bf{v}}_j(t)$.}

Meanwhile, the overall radio-over-fiber in-phase/quadrature (I/Q) fronthaul capacity for the $n$-th RRH is constrained
by a capacity threshold $F_n$. Considering the
compression techniques over the fronthaul, the fronthaul capacity is not exactly equal to the accumulated data rates but rather is a utility function of it \cite{2}. A utility function should be used not only to present the linear relationship between the radio-over-fiber I/Q fronthaul capacity and the accumulated data rate, but also to incorporate the impact of the compression technique. Therefore, the fronthaul capacity constraint is expressed as
\begin{align}
&g\{\sum\limits_{k = 1}^{K_R} \mathds{{1}} \{
{\bf{v}}_{n,k}^H(t){\bf{v}}_{n,k}(t) \} R_k(t)\} \nonumber\\
=&g\{\sum\limits_{k =
1}^{K_R} \mathds{1} \{
{\bf{v}}_{k}^H(t){\mathbf{{D}}_n^H}{\mathbf{{D}}_n}{\bf{v}}_{n}(t)
\} R_k(t)\} \leq F_n, n \in \cal{N},\label{eq:fronthaul}
\end{align}
where $g(\cdot)$ reflects the relationship
between the accumulated data rate of radio access links and the radio-over-fiber I/Q fronthaul capacity under a given compression technique. {\color{red}Here, $\mathds{1} \{x\}$ denotes the indicator function of set $R/\{0\}$ for $x \geq 0$}:
\begin{equation}
\mathds{1}\{ x\} = \left\{ \begin{array}{l}
0, \quad if \quad x = 0\\
1, \quad else
\end{array} .\right.
\end{equation}

Let $C_n = g^{-1}(F_n)$; then the radio-over-fiber I/Q fronthaul capacity constraint can be expressed equivalently that the accumulated data rate is not beyond $C_n$, which can be written as
\begin{equation}
\begin{split}
\sum\limits_{k = 1}^{K_R}\! \mathds{{1}} \{\!
{\bf{v}}_{n,k}^H(t){\bf{v}}_{n,k}(t) \}\! R_k(t) &\!=\!\sum\limits_{k =
1}^{K_R}\! \mathds{1} \{
{\bf{v}}_{k}^H(t){\mathbf{{D}}_n^H}{\mathbf{{D}}_n}{\bf{v}}_{n}(t)
\} R_k(t) \\
&\leq C_n, n \in \cal{N}.\label{eq:fronthaul}
\end{split}
\end{equation}

When considering all constraints, including the average power
consumption of each RRH expressed in (\ref{eq:I5}), the queue
stability expressed in (\ref{eq:stability}), the inter-tier
interference to MUEs expressed in (\ref{eq:interference}), and the
individual fronthaul capacity of each RRH expressed in
(\ref{eq:fronthaul}), the maximization of the averaged weighted EE
utility objective function for RRHs in H-CRANs can be formulated as
the following stochastic optimization problem:

\begin{equation}
\label{eq:I7}
\begin{split}
\mathop {\max }\limits_{\{ {{\bf{v}}_k(t)}\} } \quad & \overline {{\eta}_{EE}}= \mathop {\lim }\limits_{t \to \infty } \frac{1}{t}\sum\limits_{\tau = 0}^{t - 1} {\mathbb{E}\{{\eta_{EE}}(\tau )\}}\\
s.t. \quad & \\
C1: \quad & {\overline P _n} \le P_n^{avg},n \in \cal{N},\\
C2: \quad & \mathop {\lim}\limits_{t \to \infty}{\frac{\mathbb
{E}\{|Q_k(t)|\}}{t}}=0, k \in {\mathcal{K}_R},\\
C3: \quad & {P_n}(t) \le P_n^{\max },n \in \cal{N},\\
C4: \quad & \sum\limits_{j = 1}^{{K_R}} {{\bf{v}}_j^H(t){\mathbf{{g}}}_k(t){{\bf{g}}_k^H(t)}{{\bf{v}}_j(t)}} \le {\varphi _k},k \in {\mathcal{K}_M}, \forall t, \\
C5: \quad & \sum\limits_{k = 1}^{K_R} \mathds{{1}} \{
{\bf{v}}_{k}^H(t){\mathbf{{D}}_n^H}{{\mathbf{D}}_n}{\bf{v}}_{n}(t)
\} R_k(t) \leq C_n, n \in \cal{N}.
\end{split}
\end{equation}

In (\ref{eq:I7}), the constraint $C1$ ensures the long-term energy
consumption of the $n$-th RRH under the predefined level where
$P_n^{avg}$ denotes the average power consumption threshold. $C2$ is
the network stability constraint to guarantee a finite queue length
for each queue. $C3$ is the energy-saving constraint for the $n$-th
RRH where $P_n^{\max }$ denotes the maximum transmit power of the
$n$-th RRH. $C4$ is the constraint on interference from RRHs to
MUEs, and $C5$ is the constraint on the fronthaul consumption for
the $n$-th RRH.

Intuitively, the optimization objective function expressed in
(\ref{eq:I7}) with so many constraints is complex and cannot be directly
solved. We also note that $C1$ and $C2$ in (\ref{eq:I7})
are constraints on time averaged variables. Hence, they can be
satisfied only if the BBU pool has the CSI and knowledge of queue
backlogs at all time slots instantaneously, which is infeasible and
unpractical.
{\color{red}Fortunately, with Lyapunov optimization tool\textcolor[rgb]{1.00,0.00,0.00}{\cite{II:R1}},
the time-averaged constraints $C1$ and $C2$ can be transferred into instantaneous constraints,
and the optimization function with the $C1$ and $C2$ constraints can be
transformed into a queue mean-rate stable problem, which can be
solved only based on the observed CSI and queue backlogs at each
time slot.}

\section{Delay-aware EE Maximization}

In this section, the optimization problem in (\ref{eq:I7}) is
reformulated as an equivalent sum-MSE minimization problem. Then a
corresponding dynamic network-wide beamforming algorithm is proposed.

\vspace{-10pt}
\subsection{General Lyapunov Optimization}

Before presenting the solution of the problem, we first give the following lemma to show how
the average constraint can be transformed into a queue stability problem.

\begin{lemma}
Construct a virtual queue $H_n(t)$, the queue dynamics of which are
\begin{equation}
\label{eq:II8}
{H_n}(t + 1) = {\left\{ {{H_n}(t) - P_n^{avg} +
{P_n}(t)} \right\}^ + }.
\end{equation}

Suppose $\mathbb{E}\{H_n(0)\}<\infty$, if the virtual queue $H_n(t)$ is
mean-rate stable, the inequality $\overline {P}_n \leq P_n^{avg}$ can be satisfied.
\end{lemma}

\begin{proof}
Suppose that $H_n(t)$ is mean-rate stable; then we have
$\mathop {\lim}\limits_{t \to \infty}{\frac{\mathbb
{E}\{|H_n(t)|\}}{t}}=0$ with the probability 1 based on Definition 1.
Summing (\ref{eq:II8}) from $t=0$ to $T-1$, we have that
\begin{equation}
\begin{split}
\sum\limits_{t = 0}^{T - 1}{H_n}(t + 1) =& \sum\limits_{t = 0}^{T - 1}{\left\{ {{H_n}(t) - P_n^{avg} +
{P_n}(t)} \right\}^ + } \\
= & \sum\limits_{t = 0}^{T - 1}\{ {H_n}(t) - P_n^{avg}  \\
& +\max \{{P_n}(t),P_n^{avg}-H_n(t)\} \}  \\
H_n(T) - H_n(0) = & \sum\limits_{t = 0}^{T - 1} {\max \{{P_n}(t),P_n^{avg}-H_n(t)\}}
-TP_n^{avg} \\
 \geq & \sum\limits_{t = 0}^{T - 1} {P_n(t)}
-TP_n^{avg},
\end{split}
\end{equation}
holds for all time slots $t > 0$. Taking the sum-expectation operation and the limit as $T \to \infty$,
we can conclude that
\begin{equation}
\mathop {\lim}\limits_{T \to \infty}{\frac{\mathbb
{E}\{|H_n(T)|\}}{T}} - \mathop {\lim}\limits_{T \to \infty}{\frac{\mathbb
{E}\{|H_n(0)|\}}{T}} =0 \geq \overline {P}_n - P_n^{avg}.
\end{equation}

Therefore, $\overline {P}_n \leq P_n^{avg} $ holds.
\end{proof}

With \emph{Lemma 1}, the constraint $C1$ in (11) is
transformed into a queue stability problem by constructing a virtual
queue $H_n(t)$ for each RRH $n$.

With the actual queues (\ref{eq:I3}) and virtual queues (\ref{eq:II8}),
denote ${\bf{\Theta }}(t) = \left[ {{\bf{Q}}(t),{\bf{H}}(t)}
\right]$ as the combined matrix of all the actual and virtual
queues, where ${\bf{H}} (t)=\left\{ \left. {H_n (t)} \right| n =
1,...,N \right\}$, the Lyapunov function is defined as a scalar
metric of queue congestion:
\begin{equation}
\label{eq:II9} L\left( {{\bf{\Theta }}(t)} \right) \buildrel \Delta
\over =\frac{1}{2}\left\{ {\sum\limits_{k = 1}^{{K_R}}
{{Q_k}^2{(t)}}
 + \sum\limits_{n = 1}^N {{H_n}^2{(t)}} } \right\}.
\end{equation}

The Lyapunov drift is introduced to push the Lyapunov function to a
lower congestion state and keep queues stable, which is defined as
\begin{equation}
\label{eq:II10} \Delta \left( {{\bf{\Theta }}(t)} \right) \buildrel
\Delta \over = \mathbb{E} \left\{ {L\left( {{\bf{\Theta }}(t + 1)}
\right)
 - L\left( {{\bf{\Theta }}(t)} \right)} \right\}.
\end{equation}

In terms of Lyapunov optimization, the underlying objective of
optimal network-wide beamformer design is to minimize an infimum bound on the
drift-plus-penalty expression in each time slot:
\begin{equation}
\label{eq:II11} \Delta\! \left( {{\bf{\Theta }}(t)} \! \right)\! -\!
V\!\mathbb{E} \left\{ {{\eta}_{EE}(t)|{\bf{\Theta }}(t)} \right\},
\end{equation}
where $V$ is a non-negative parameter controlling the tradeoff
between the average weighted EE performance and average queue
backlog. For simplicity, this parameter is termed as EE-delay tradeoff
in the following discussions.

With the queue dynamics of $Q_k(t)$ and $H_n(t)$ presented in
(\ref{eq:I3}) and (\ref{eq:II8}), respectively, and the definition
of Lyapunov drift in (\ref{eq:II10}), the following lemma holds.

\begin{lemma} At any time slot $t$, with the observed queue state
and CSI, and under any network-wide beamformer control decision, the
drift-plus-penalty satisfies the following inequality:
\begin{equation}
\label{eq:II12}
\begin{split}
& \Delta \left( {{\bf{\Theta }}(t)} \right) - V \mathbb{E}\left\{ {{\eta}_{EE}(t)|{\bf{\Theta }}(t)} \right\} \\
\le& B - V \mathbb{E} \left\{ {\eta_{EE}(t)|{\bf{\Theta }}(t)} \right\} \\
& + \sum\limits_{n = 1}^N {{H_n}(t) \mathbb{E} \left\{ {{P_n}(t) -
P_n^{avg}|{\bf{\Theta }}(t)} \right\}} \\
&+ \sum\limits_{k = 1}^{{K_R}}
{{Q_k}(t) \mathbb{E} \left\{ {{A_k}(t) - {R_k}(t)|{\bf{\Theta }}(t)}
\right\}},
\end{split}
\end{equation}
where $B>0$ is a finite constant which satisfies (\ref{eq:II13}) for $\forall t$ :
\begin{equation}
\label{eq:II13}
\begin{split}
B \ge & \ \frac{1}{2}\mathbb{E}\left\{ {\sum\limits_{k =
1}^{{K_{\rm{R}}}} {\left( {{A_k^2}{{(t)}} + {R_k^2}{{(t)}}} \right)}
} {|{\bf{\Theta }}\left( t \right)}\right\} \\
&+\frac{1}{2}\mathbb{E}\left\{ {\sum\limits_{n = 1}^N {{{\left(
{{P_n}(t) - P_n^{avg}(t)} \right)}^2}} } {|{\bf{\Theta }}\left( t
\right)}\right\}.
\end{split}
\end{equation}
\end{lemma}

\begin{IEEEproof}
See Appendix A.
\end{IEEEproof}

To push the underlying objective (\ref{eq:II11}) to its minimum, a
proper network-wide beamformer ${\bf{v}}_k(t)$ is chosen to
greedily minimize the drift-plus-penalty in (\ref{eq:II11}). As a result, a strategy is proposed herein
to minimize the right-hand-side (R.H.S) of the inequality of drift-plus-penalty in (\ref{eq:II12}) based on the observed QSI and
CSI at each time slot $t$ instead of minimizing (\ref{eq:II11}) directly. Based on the concept of opportunistically
minimizing an expectation, this is accomplished by greedily
minimizing as follows:
\begin{equation}
\label{eq:II14} \mathop {\min }\limits_{\{ {{\bf{v}}_k(t)}\} }
\left[ {\sum\limits_{n = 1}^N {{H_n(t)}{P_n(t)}} - \sum\limits_{k =
1}^{{K_R}} {{Q_k(t)}{R_k(t)}} - V \eta_{EE}(t)} \right].
\end{equation}

For notational simplicity, $X_n(t)$ and $Y_k(t)$ are denoted by
\begin{equation}
\label{eq:II15}
\begin{split}
{X_n}(t) = & \ {H_n(t)} + \frac{V(1 - \alpha ){\mu_n}}{N},\\
{Y_k}(t) = & \ {Q_k(t)} + \frac{V\alpha {\omega _k}}{K_R},
\end{split}
\end{equation}
respectively. With the constraints $C1$ and $C2$ incorporated into
the objective function (\ref{eq:II14}), the optimization problem
(\ref{eq:I7}) can be rewritten as
\begin{equation}
\label{eq:II16}
\begin{split}
\mathop {\min }\limits_{\{ {{\bf{v}}_k(t)}\} } \quad & \sum\limits_{n = 1}^N {{X_n(t)}{P_n(t)}} - \sum\limits_{k = 1}^{{K_R}} {{Y_k(t)}{R_k(t)}} \\
s.t. \quad & C3, C4, C5,
\end{split}
\end{equation}
where $X_n(t)$ and $Y_k(t)$ can be calculated by the observed QSI at
time slot $t$. $P_n(t)$ and $R_k(t)$ are based on the CSI at
time slot $t$ and the beamforming vector ${\bf{v}}_k(t)$. Thus, a
network-wide cooperative beamformer design algorithm is proposed in
the next subsection to solve problem (\ref{eq:II16}).

\vspace{-10pt}
\subsection{Beamformer Design Algorithm}

The optimization problem (\ref{eq:II16}) is non-convex,
which is difficult to solve directly.
To present a solution of
(\ref{eq:II16}), the performance of control actions to obtain a
local minimum which is within an additive constant of the infimum is analyzed.
Therefore, in the following content, the definition of
\emph{C-additive approximation}\textcolor[rgb]{1.00,0.00,0.00}{\cite{II:R1}} is first introduced, based on
which the locally optimal solution of problem (\ref{eq:II16}) is
analyzed. Finally, a network-wide beamformer design algorithm is
proposed.

\begin{Definition}

For a given constant $C \geq 0$, a \emph{C-additive approximation}
of the drift-plus-penalty algorithm is to choose an action that
yields a conditional expected value on the right-hand-side of the
drift-plus-penalty (given $\boldsymbol{\Theta}(t)$) at time slot
$t$, that is within a constant $C$ from the infimum over all
possible control actions.
\end{Definition}

Based on \emph{Definition 3}, a local optimal beamformer algorithm
can be designed. Note that in the constraint $C5$, the indicator
function can be equivalently expressed as a scalar $\ell0$-norm,
which is the number of nonzero entries in a vector. The indicator
function can be approximated by a convex re-weighted
$\ell1$-norm\textcolor[rgb]{1.00,0.00,0.00}{\cite{II:ECand}\cite{Yuwei1}}, i.e.,

\begin{equation}
\mathds{{1}}
\{{{\bf{v}}_{k}^H(t)}{\mathbf{{D}}_n^H}{\mathbf{{D}}_n}{\bf{v}}_{k}(t)\}
={\beta_n^k(t)} {{\bf{v}}_{k}^H(t)}{\mathbf{{D}}_n^H}{\mathbf{{D}}_n}{{\bf{v}}_{k}(t)},
\end{equation}
where $\beta_n^k(t)$ is updated iteratively according to $\beta_n^k(t) =
\frac{1}{{{\bf{v}}_{k}^H(t)}{\mathbf{{D}}_n^H}{\mathbf{{D}}_n}{{\bf{v}}_{k}(t)}+\kappa}$
with a small constant regularization factor $ \kappa> 0$ and
${\bf{v}}_{k}$ from the previous iteration. Since it is still difficult to solve a problem that
involves $R_k(t)$ in both the objective function and the constraints, an iterative scheme is used
in which the fixed value of $R_k(t)$ from the previous iteration is adopted here. Thus, the optimization
problem can be written as
\begin{equation}
\label{eq:II17}
\begin{split}
\mathop {\min }\limits_{\{ {{\bf{v}}_k(t)}\} } \quad & \sum\limits_{n = 1}^N {{X_n(t)}{P_n(t)}} - \sum\limits_{k = 1}^{{K_R}} {{Y_k(t)}{R_k(t)}} \\
s.t. \quad & C3, C4,\\
\quad & C6: \sum\limits_{k=1}^{K_R}{\beta_n^k}
{{\bf{v}}_{k}^H(t)}{\mathbf{{D}}_n^H}{\mathbf{{D}}_n}{{\bf{v}}_{k}(t)}
\tilde{R}_k(t)\le C_n,
\end{split}
\end{equation}
where $\tilde{R}_k(t)$ in $C6$ is the rate of the previous
iteration. {\color{red}Obviously, the approximated problem (\ref{eq:II17}) is
still non-convex, while it can be reformulated as an equivalent WMMSE
problem to achieve a local optimum via the \emph{C-additive approximation} of the
drift-plus-penalty algorithm.} Inspired by the equivalence between
the weighted sum rate (WSR) maximization and WMMSE\textcolor[rgb]{1.00,0.00,0.00}{\cite{III:R11,III:R12,III:R13}} for the MIMO channel,
the generalized WMMSE equivalence established in\textcolor[rgb]{1.00,0.00,0.00}{\cite{III:R12}} is
extended to solve the problem (\ref{eq:II17}) in the H-CRAN
scenario. We state this equivalence as follows.

\begin{proposition}
The problem (\ref{eq:II17}) has the same optimal solution as the following WMMSE problem:
\begin{equation}
\label{eq:II18}
\begin{split}
\mathop {\min } \limits_{\{w_k(t),u_k(t),{\bf{v}}_{k}(t)\}} \quad & \sum\limits_{k = 1}^{{K_R}} {{Y_k(t)}\left\{ {{w_k(t)}{e_k(t)} - \log {w_k(t)}} \right\}} \\
&+ \sum\limits_{n = 1}^N {{X_n(t)}\sum\limits_{k = 1}^{K_R} {{\bf{v}}_k^H(t){\mathbf{{D}}_n^H}{\mathbf{{D}}_n}{{\bf{v}}_k(t)} }} , \\
 s.t. \quad & C3, C4, C6.
\end{split}
\end{equation}
where $w_k(t)$ denotes the mean-square error (MSE) weight for user $k$ at time
slot $t$, and $e_k(t)$ is the corresponding MSE defined as
\begin{equation}
\label{eq:II19}
\begin{split}
{e_k(t)} \buildrel \Delta \over = & \ {\mathbb{E}}\left\{ {{{\left( {{u_k(t)}{y_k(t)} - {s_k(t)}} \right)}^2}} \right\} \\
= &u_k^H(t) \Big( \sum\limits_{j = 1}^{{K_R}} {{\bf{v}}_j^H(t){\bf{h}}_k(t){{\bf{h}}_k^H(t)}{{\bf{v}}_j(t)}} \Big) u_k(t) \\
 & + u_k^H(t) \Big( \sum\limits_{i = 1}^{{K_M}} {{\bf{v}}_{0,i}^H(t){\bf{g}}_{0,k}(t){{\bf{g}}_{0,k}^H(t)}{{\bf{v}}_{0,i}(t)}} \Big) u_k(t) \\
 &- 2{\bf{Re}}\{{u_k(t)}{{\bf{h}}_k^H(t)}{{\bf{v}}_k(t)}\} + {\sigma ^2}{\bf{Re}}\{u_k^H(t)u_k(t)\} + 1,
\end{split}
\end{equation}
under the receiver $u_k(t) \in \mathbb{C}$.
\end{proposition}

Based on the equivalent WMMSE problem (\ref{eq:II18})
which is convex with respect to each of the individual optimization
variables, the averaged weighted EE utility objective maximization
problem (\ref{eq:I7}) can be solved. This crucial observation allows
the problem (\ref{eq:I7}) to be solved efficiently through the block
coordinate descent method by iterating among
${{\bf{v}}_k(t)},{u_k}(t)$, and ${w_k(t)}$:
\begin{itemize}
\item[\emph{$\bullet$}] The optimal MSE weight ${w_k(t)}$ under the fixed ${{\bf{v}}_k(t)}$ and ${u_k(t)}$ is
given by
\begin{equation}
\label{eq:II20} w_k^{opt}(t) = \ e_k^{ - 1}(t).
\end{equation}

\item[\emph{$\bullet$}] The optimal receiver ${u_k(t)}$ under the fixed ${{\bf{v}}_k(t)}$ and ${w_k(t)}$ is
given by
\begin{equation}
\label{eq:II22}
\begin{split}
u_k^{opt}(t)\! =&\! {{\bf{h}}_k^H(t)}{{\bf{v}}_k(t)} \Big\{
\sum\limits_{j = 1}^{{K_R}}\!
{\bf{v}}_j^H(t){\bf{h}}_k(t){{\bf{h}}_k^H(t)}{{\bf{v}}_j(t)}\! +\!\phi(t) \Big\}^{-1}.
\end{split}
\end{equation}

\item[\emph{$\bullet$}] The optimization problem for finding the optimal transmit
network-wide beamformer ${{\bf{v}}_{k}(t)}$ under the fixed
${w_k(t)}$ and ${u_k(t)}$ is

\begin{equation}
\label{eq:II21}
\begin{split}
\mathop {\min }\limits_{\{ {{\bf{v}}_k(t)}\} } \quad & \sum\limits_{k = 1}^{{K_R}} {\bf{v}}_k^H(t)\Big( \sum\limits_{j = 1}^{{K_R}} {Y_j(t)}{w_j(t)}u_j^H(t){\bf{h}}_j(t)\\
&{{\bf{h}}_j^H(t)}u_j(t) + \sum\limits_{n = 1}^N {{X_n(t)}{\mathbf{{D}}}_n^H{{\mathbf{{D}}}_n}} \Big){{\bf{v}}_k(t)} \\
& - 2\sum\limits_{k = 1}^{{K_R}} {{Y_k(t)}{w_k(t)}
{\bf{Re}}\{{u_k(t)}{{\bf{h}}_k^H(t)}{{\bf{v}}_k(t)}\}} \\
s.t. \quad & C3, C4, C6.
\end{split}
\end{equation}

The optimization objective function in (\ref{eq:II21}) is a
quadratically constrained quadratic programming (QCQP) problem and
can be solved using a standard convex optimization solver such as
the Matlab software for disciplined convex programming (CVX)\textcolor[rgb]{1.00,0.00,0.00}{\cite{CVX}}.
\end{itemize}

The above proposed WMMSE approach for solving the original optimization
problem (\ref{eq:I7}) can be summarized in \textbf{Algorithm
\ref{alg:1}}.

{\color{red}
\begin{algorithm}[htb]
\caption{Averaged weighted EE maximization with per-RRH power
and interference constraints at time slot $t$.}\label{alg:1}
\begin{algorithmic}[1]
\REQUIRE {\color{red}Initial network-wide beamforming vector ${\bf{v}}_k(t)$ and corresponding
$\eta_{EE}(t)$, and the precision $\kappa$;}
\ENSURE Calculate the optimal ${\bf{v}}_k^*(t)$ and corresponding
$\eta_{EE}^*(t)$.
\REPEAT
\STATE {\color{red}Update ${\bf{v}}_k^*(t)={\bf{v}}_k(t)$ and ${\eta}^*_{EE}(t)={\eta}_{EE}(t)$;}
\STATE With ${\bf{v}}_k(t)$ fixed, compute
the MMSE receiver $u_k(t)$ and the corresponding MSE $e_k(t)$
according to (\ref{eq:II22}) and (\ref{eq:II19});
\STATE Update the
MSE weight $w_k(t)$ according to (\ref{eq:II20});
\STATE Find the
optimal transmit network-wide beamformer ${\bf{v}}_k(t)$ under fixed
$u_k(t)$ and $w_k(t)$, by solving the QCQP problem (\ref{eq:II21});
\STATE Compute the achievable rate $R_k(t)$ and energy consumption
$P_n(t)$ according to (\ref{eq:I2}) and (\ref{eq:I5}), respectively;
\STATE Compute the EE function ${\eta}_{EE}(t)$; \STATE {\color{red}Update
$\beta_{n}^{k}(t)$ and $\tilde{R}_k(t)$. }
\UNTIL {\color{red} $|\eta_{EE}^*(t)-\eta_{EE}(t)|\leq \kappa|\eta_{EE}^*(t)|$.}
\end{algorithmic}
\end{algorithm}
}

{\color{red}
Note that \textbf{Algorithm 1} cannot be guaranteed to converge to the global optimum, while it can quickly converge to a local optimum. The random initial point can approach a local optimum through \textbf{Algorithm 1} with a substantial number of iterations. To decrease the number of iterations and approach a local optimal solution quickly, it
is critical to choose proper initialization points with reasonable approaches such as the interference
alignment initialization proposed in \cite{III:R15}.
}

Remark 2: \textbf{Algorithm 1} is based on the block
coordinate descent method. In this case, the computational complexity
of Step 2 in \textbf{Algorithm 1} is $O(K_R^2NL_R)$,
mainly due to the receive covariance matrix computation in
(\ref{eq:II19}) and (\ref{eq:II22}). With the MSE $e_k(t)$ obtained from Step 2, the
additional computational complexity for Step 3 to update
all MSE weights $w_k(t)$ is only $O(K)$. Step 4 requires
solution of the QCQP problem, which is the largest part of the computational complexity in \textbf{Algorithm 1}.
The total number of variables in the QCQP problem is
$K_RL_RN$ and the computation complexity of using the CVX
method to solve such an QCQP problem is approximately
$O((K_RL_RN)^{3.5})$. Step 5 has the same computational complexity as computing the MSE, and the
computational complexity of the remaining steps is $O(K)$.

It is noted that the complexity of solving such a
QCQP problem is related to the number of potential transmit antennas
${L_R}N$ serving each user and the total number of users $K_R$ to be
considered in each iteration. Thus, to improve the efficiency of
\textbf{Algorithm 1} in each iteration, it is practical to
reduce the number of potential transmit antennas ${L_R}N$ and the
total number of users $K_R$. This can be done by iteratively removing the $n$-th
RRH from the $k$-th RUE's candidate cluster once the transmit power
from the $n$-th RRH to the $k$-th RUE, i.e.,
${{\bf{v}}_{n,k}^H{{\bf{v}}_{n,k}}}$, is below a certain threshold,
or checking the achievable RUE rate $R_k$ and energy consumption
$P_n$ iteratively and ignoring those RUEs with negligible rates
during the next iteration. Such solutions can reduce the number
of needed iterations and decrease the complexity of
\textbf{Algorithm 1}.

\section{Performance Analysis}

In this section, the average performance of the weighted EE utility
function and the queue length bound achieved by the proposed
\textbf{Algorithm \ref{alg:1}} are introduced, which leads to an
EE-delay tradeoff.

Before presenting the theoretical performance of the proposed
\textbf{Algorithm \ref{alg:1}}, we introduce the following
assumption under the given channel condition and
network-wide beamformer design algorithm:
\begin{align}
&R_k(t) \le {R_k^{max}}, \mathbb{E}\{R_k^2(t)\} \le ({R_k^{max}})^2,\label{eq:V0}\\
&\mathbb{E}{\{P_n(t)-P_n^{av}\}} \le \gamma,\label{eq:V1}\\
&\mathbb{E}{\{\eta_{EE}(t)\}} \le {\eta_{EE}^{max}},\label{eq:V2}
\end{align}
where $R_k^{max}$, $\gamma$ and $\eta_{EE}^{max}$ are finite
constants. The assumptions (\ref{eq:V0}) and (\ref{eq:V1}) are
reasonable in realistic systems since the data rate of each RUE and
the power allocation of each RRH are constrained. The assumption
(\ref{eq:V2}) strictly depends on the boundedness assumptions of
(\ref{eq:V0}) and the power constraints.

\begin{Definition}
\textit{The network capacity region is the set ${\Lambda}$ of
non-negative rate vectors $\boldsymbol {\lambda}$, satisfying:}
\begin{equation}
\label{eq:IV23} {\epsilon}_{max}({\boldsymbol {\lambda}}) \geq 0 ,
\end{equation}
\textit{where ${\epsilon}_{max} ({\boldsymbol {\lambda}})$ is the maximum value of ${\epsilon}$ which satisfies ${\lambda}_k+{\epsilon} \le \overline {R}_k,k \in {\cal {K}}_R$.}
\end{Definition}

\begin{lemma} Suppose that $\boldsymbol{\lambda}$ is strictly interior to the capacity
 region $\Lambda $, and let $\epsilon$ be a positive value such that
 $\epsilon < \epsilon_{max}(\boldsymbol{\lambda})$. If constraints are feasible,
 then for any $\theta >0 $, there exists an algorithm that makes independent,
 stationary and randomized decisions about the network-wide beamformer at each time slot based only
 on the observed network state, which satisfies
\begin{equation}
\label{eq:IV24}
\begin{split}
\mathbb{E}\left\{ {{A_k}(t) - R_k^*(t)|{\bf{\Theta }}(t)} \right\} &= \mathbb{E}\left\{ {{A_k}(t) - R_k^*(t)} \right\} \le - \epsilon, \\
\mathbb{E}\left\{ {{P_n}(t) - P_n^{av}|{\bf{\Theta }}(t)} \right\} &\le \theta,\\
\mathbb{E}\left\{ {{\eta}_{EE}^*(t)|{\bf{\Theta }}(t)} \right\} &\ge
{\eta}_{EE}^{opt} - \theta,
\end{split}
\end{equation}
where ${\eta}_{EE}^*(t)$, $y_n^*(t)$ and $R_m^*(t)$ are
corresponding results under the stationary algorithm, and
${\eta}_{EE}^{opt}$ is the theoretical optimal value of
$\overline{\eta_{EE}}$ under constraints \emph{C1}, \emph{C2},
\emph{C3}, \emph{C4} and \emph{C5} in (\ref{eq:I7}).
\end{lemma}

The detailed proof for \emph{Lemma 3} are omitted for
simplicity as a similar proof can be found in\textcolor[rgb]{1.00,0.00,0.00}{\cite{II:R1}}.

\vspace{-10pt}
\subsection{Stability of Queues}

In Section III, we proposed a beamformer design algorithm utilizing
Lyapunov optimization technique, and the constraints $C1$, $C2$ are
incorporated into the process of Lyapunov optimization. \emph{Theorem 1}
shows that constraints $C1$, $C2$ are guaranteed under the proposed \textbf{Algorithm \ref{alg:1}}.

\begin{theorem}
Suppose that $\mathbb{E}\{L(\boldsymbol{\Theta}(0))\}<\infty$ and the multimedia traffic arrival
rate $\boldsymbol{\lambda}$ is in within the network capacity region, then under the proposed
beamformer design algorithm, all the actual queues and virtual queues are mean-rate stable.
\end{theorem}

\begin{IEEEproof}
See Appendix B.
\end{IEEEproof}

\emph{Theorem 1} shows that constraints $C1$ and $C2$ are satisfied under \textbf{Algorithm \ref{alg:1}}
according to \emph{Definition 1} and \emph{Lemma 1}, respectively.

\vspace{-10pt}
\subsection{Average weighted EE Performance}

The average weighted EE performance obtained by \textbf{Algorithm \ref{alg:1}} is given by \emph{Theorem 2}.

\begin{theorem}Utilizing the proposed optimization solution, for any $V>0$, the gap
between the average weighted EE and the optimum is within
$\mathcal{O}(1/V)$:
\begin{equation}
\label{eq:IV25}
\overline {{\eta}_{EE}} \ge {\eta}_{EE}^{opt} - \frac{B+C}{V},
\end{equation}
\end{theorem}
where $C$ is a constant gap between the local optimum obtained by the
proposed dynamic network-wide beamformer design algorithm and the
infimum of the R.H.S of (\ref{eq:II12}).

\begin{IEEEproof}
See Appendix C.
\end{IEEEproof}

As the optimal solution of the optimization problem (\ref{eq:I7}) with different non-trivial constraints is difficult to obtain in
practice. \emph{Theorem 2} suggests that an near-to-optimal solution can be obtained arbitrarily close to the optimum
by adjusting the control parameter $V$. That is, if $V$ is sufficiently large, the
average weighted EE performance can be pushed arbitrarily close to
the optimum, which is more realistic than attempting to achieve the optimal with high complexity.

\vspace{-10pt}
\subsection{Queueing Bounds}

To evaluate the constant queuing bound of the average queue length
for the proposed dynamic network-wide beamformer design algorithm,
the following theorem can be used.

\begin{theorem} Assume that the network-wide beamformer of each RRH and the queue dynamics are determined by
the proposed dynamic \textbf{Algorithm \ref{alg:1}}. Then, the average
queue bound satisfies
\begin{equation}
\label{eq:IV30}
\overline Q \! = \! \mathop {\lim }\limits_{T \to \infty } \frac{1}{T}\sum\limits_{t = 0}^{T-1} {\sum\limits_{k = 1}^{{K_R}} {\mathbb{E}\left\{ {{Q_k}(t )} \right\}} } \le \frac{{B \!+ \! C \! + \!V \left( { {\eta}_{EE}^{max} }\!-\!{\eta}_{EE}^{opt} \right)}}{\epsilon }.
\end{equation}
\end{theorem}

\begin{IEEEproof}
See Appendix D.
\end{IEEEproof}

\emph{Theorem 3} shows that the average queue length is bounded by a
deterministic upper backlog bound which increases linearly with $V$.

\begin{remark}
\emph{Theorem 2} combined with \emph{Theorem 3} show that the
proposed dynamic network-wide beamformer design algorithm achieves
an $[\mathcal{O}(1/V),\mathcal{O}(V)]$ tradeoff between the average
weighted EE performance and queue backlogs, which leads to an
EE-delay tradeoff for a given arrival rate according to Little's
Theorem\textcolor[rgb]{1.00,0.00,0.00}{\cite{IV:R14}}. With an increase of the control parameter
$V$, the achieved weighted EE performance becomes better at the cost
of incurring the larger queuing delay. Therefore, it is important to
choose a proper $V$ to obtain the required performance and QoS in
realistic H-CRANs .
\end{remark}

\section{Simulation Results and Analysis}

{\color{red}To evaluate the performance of the proposed optimal network-wide
cooperative beamformer design algorithm in H-CRANs, we consider one radio resource block for all RRHs and MBSs. When more radio resource blocks are used for each RRH and MBS, the radio resource block allocation algorithms that adapt to the time-varied radio channel and dynamic traffic arrival described in {\cite{I:A5,I:A6,I:A7}} can be directly used along with the proposal in this paper. To decrease the high complexity and reduce the simulation time, a small-scale H-CRAN system
consisting of one MBS and 2 RRHs is considered with the assumed
simulation parameters shown in Table I. Since only one radio resource block is considered, only one RUE can be served in each RRH. As a result, it is assumed that $K_R = 4$ and $K_M = 4$. Note that similar simulation results to those presented below can be achieved for a large-scale H-CRAN system consisting of more RRHs and MBSs.}

\begin{table}[h]
\centering \caption{Simulation Parameters}\vspace{-10pt}
\small
\begin{tabular}{|m{5cm}<{\centering}|m{2cm}<{\centering}|}
\hline
 Num. of (MBSs, RRHs, MUEs, RUEs) & (1, 2, 4, 4)\\
\hline
 Num. of antennas $/$(MBS, RRH) & (2, 2) \\
\hline
 Noise power spectral density & -174 dBm/Hz \\
\hline
 Path loss exponent for transmission from BS to UE & 4 \\
\hline
 Small-scale fading & Rayleigh fading \\
\hline
 Maximum transmit power of RRH & 0.22W \\
\hline
 Average transmit power constraint of RRH & 0.2 W\\
\hline
 Transmit power of MBS & 20W \\
\hline
\end{tabular}
\label{TableChannel1}
\end{table}

\vspace{-10pt}
\subsection{Queue Stability Evaluation}

To evaluate the queue stability achieved by the proposed dynamic
network-wide beamformer design algorithm, we take the user queues at
the arrival rate $\lambda=${\kern 1pt}{\kern 1pt}4.2{\kern 1pt}{\kern 1pt}bit/slot/Hz as sample queues. The user
average queue length against $t$ under several $V$ is shown in
Fig. \ref{Q_t}. It can be observed that the average queue length
first increases with $t$ and then fluctuates around certain fixed
values. Larger $V$ leads to larger stable values which directly
follows with \emph{Theorem 3}.

\begin{figure}[!h!t]
\centering
\includegraphics[width=0.5\textwidth]{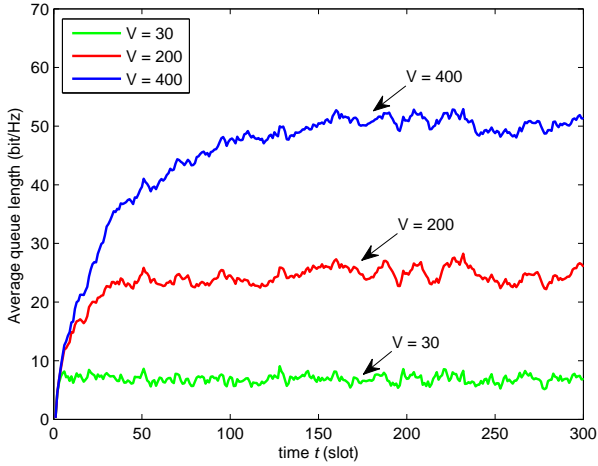}
\vspace*{-10pt} \caption{Queue length versus simulation time
length}\vspace*{-10pt} \label{Q_t}
\end{figure}

\vspace{-10pt}
\subsection{The EE-Delay Tradeoff}

The average queue backlog achieved by the proposed dynamic
network-wide beamformer design algorithm shown in Fig. \ref{Q_V}
grows linearly in $\mathcal{O}(V)$ under given multimedia traffic arrival
rate $\lambda$, which consolidates \emph{Theorem 3}. Under the same
$V$, the queue length differs when the traffic arrival rate
$\lambda$ changes, since different arrival rates cause different amounts of
power consumption. Intuitively, a larger $\lambda$ leads to higher
average power consumption since more power is required to transmit
data arrivals and avoid queuing congestion, which is supported by
Fig. \ref{P_V}.

\begin{figure}[!h!t]
\centering
\includegraphics[width=0.5\textwidth]{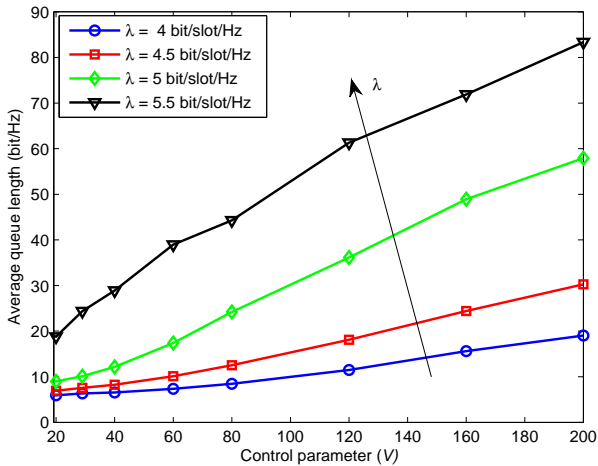}
\vspace*{-10pt} \caption{Average queue length versus
$V$}\vspace*{-10pt} \label{Q_V}
\end{figure}

Fig. \ref{P_V} is shown to evaluate the average power consumption
 as a function of $V$. It is observed that the larger the multimedia traffic arrival rate
$\lambda$, the bigger the average power consumption. This is due
to the fact that it is required for the system to consume more power
to timely transmit more multimedia traffic arrivals. Meanwhile, the
average power consumption decreases as $V$ increases for a given
$\lambda$. This can be explained by the fact that a larger $V$ implies that the system emphasizes the average weighted EE more, but an increase in transmit power does not result in a proportional increase in transmit rate due to the diminishing slope of the logarithmic rate-power function.
Therefore, it is necessary to decrease the transmit power to improve the average weighted EE performance.

\begin{figure}[!h!t]
\centering
\includegraphics[width=0.5\textwidth]{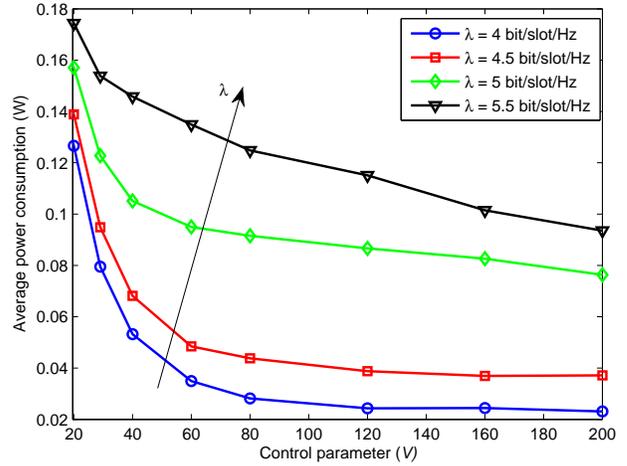}
\vspace*{-10pt} \caption{Average power consumption versus
$V$}\vspace*{-10pt} \label{P_V}
\end{figure}

In Fig. \ref{RC_V}, the average weighted EE performance versus the
parameter $V$ for different user arrival rates $\lambda$ is
evaluated to support \emph{Theorem 2}. The average weighted EE
performance increases with $V$ for any given arrival rate $\lambda$,
which can be intuitively understood by the fact that greater emphasis is placed on
the average weighted EE more when $V$ increases. The lower the multimedia traffic arrival rate $\lambda$ is, the higher the average weighted EE performance under a given control parameter $V$ will be. This happens because both transmit rate and power consumption decrease with decreasing $\lambda$ and
the the logarithmic rate-power function has the characteristic of diminishing slope .

\begin{figure}[!h!t]
\centering
\includegraphics[width=0.5\textwidth]{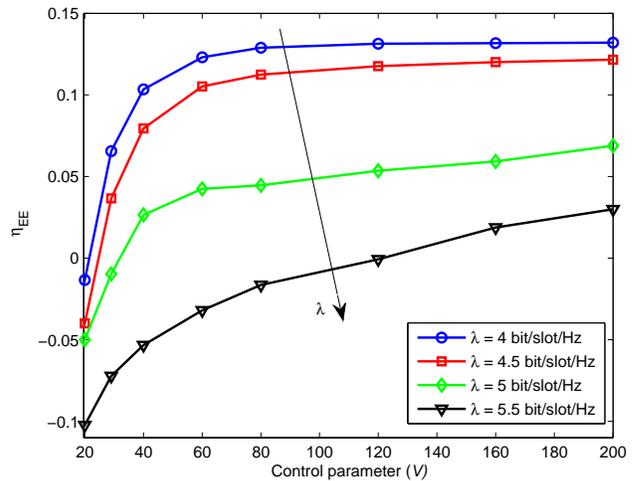}
\vspace*{-10pt} \caption{Average weighted EE performance versus
$V$}\vspace*{-10pt} \label{RC_V}
\end{figure}

To emphasize the efficiency and usefulness of the proposed
average weighted EE performance metric $\eta_{EE}$, the traditional EE
performance metric $\tilde{\eta}_{EE}$ defined in (\ref{EE_T}) is
illustrated in Fig. \ref{EE_V} as a baseline. In Fig. \ref{EE_V}, the average $\tilde{\eta}_{EE}$ grows at the rate of $O(1/V)$. The performance of $\tilde{\eta}_{EE}$ becomes better with a larger
arrival data rate $\lambda$. It can be observed that the tendency of the average performance of
$\tilde{\eta}_{EE}$ is similar to the average performance of the proposed
EE metric $\eta_{EE}$, which indicates that
the proposed EE metric $\eta_{EE}$ is valid for use as the EE measurement in H-CRANs.

\begin{figure}[!h!t]
\centering
\includegraphics[width=0.5\textwidth]{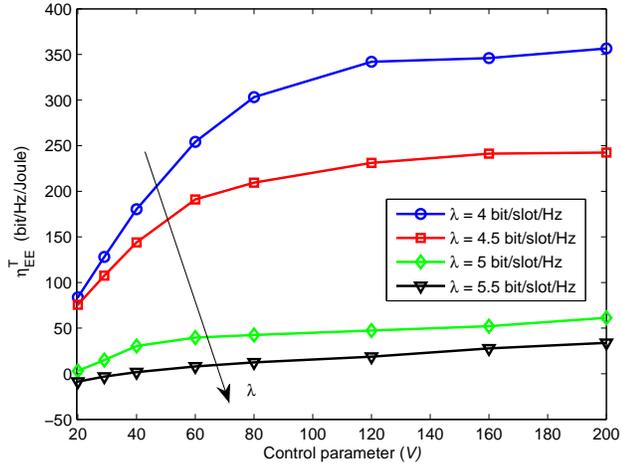}
\vspace*{-10pt} \caption{The traditional EE performance metric
versus $V$}\vspace*{-10pt} \label{EE_V}
\end{figure}

To make the tradeoff between the average weighted EE and the average
queuing delay clear, the relationship between the average
weighted EE and queuing delay is
illustrated in Fig. \ref{RC_Q} versus the parameter $V$. It can be observed that a larger $V$
leads to a better average weighted EE performance but at the cost of
incurring a larger queuing delay and vice versa. Thus, the proposed
network-wide cooperative beamformer design algorithm provides an
advanced method to flexibly balance the average weighted EE
performance and the queuing delay.

\begin{figure}[!h!t]
\centering
\includegraphics[width=0.5\textwidth]{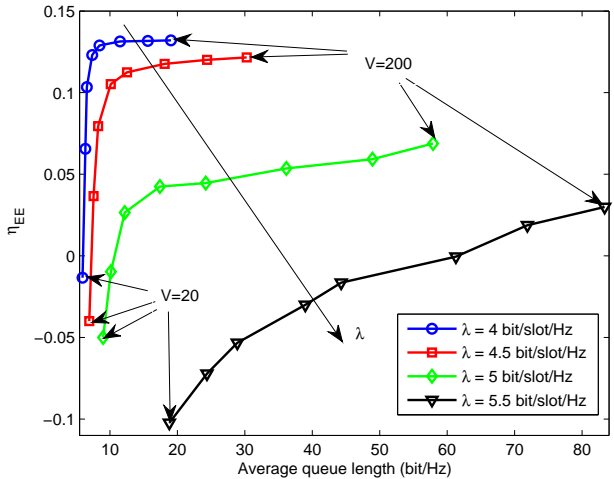}
\vspace*{-10pt} \caption{Average weighted EE performance versus
average queue length}\vspace*{-10pt} \label{RC_Q}
\end{figure}

\vspace{-10pt}
\subsection{Impact of fronthaul constraint on average weighted EE}

To evaluate the impact of the constrained fronthaul on the average
weighted EE performance in H-CRANs, we compare the proposed
network-wide cooperative beamformer design algorithm under constrained fronthual
$C_n=6$ bit/Hz with the
solution that maximizes EE with ideal fronthauls. As shown in
Fig. \ref{P_V_Fh}, compared with the fronthaul constraint under the
proposed dynamic network-wide beamformer design algorithm under the same arrival rate, the
average power consumption becomes larger if the fronthaul constraint
is not considered. This happens because the rising
average power consumption leads to an increase of
$\|{\bf{v}}_{k}\|$ in (\ref{eq:I0}).

\begin{figure}[!h!t]
\centering
\includegraphics[width=0.5\textwidth]{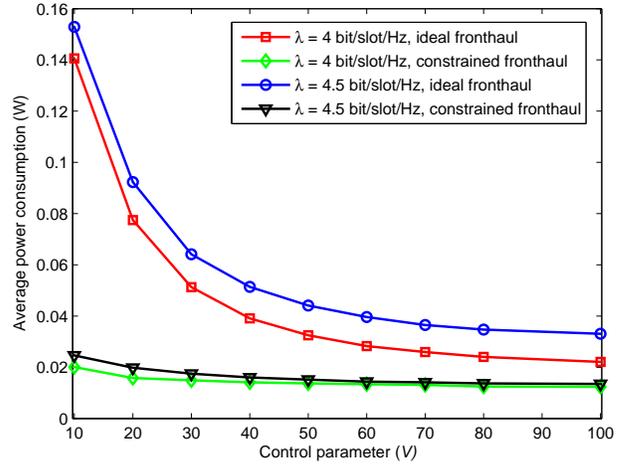}
\vspace*{-10pt} \caption{Average power consumption vs.
$V$}\vspace*{-10pt} \label{P_V_Fh}
\end{figure}

Meanwhile, under the same multimedia traffic arrival rate, the average queue
length with the fronthaul constraint is larger than without the
fronthaul constraint, which is illustrated in Fig. \ref{Q_V_Fh}. The
rational explanation is that the fronthaul constraint limits the
transmission rate and causes more congestion, which results in larger average queue length.

\begin{figure}[!h!t]
\centering
\includegraphics[width=0.5\textwidth]{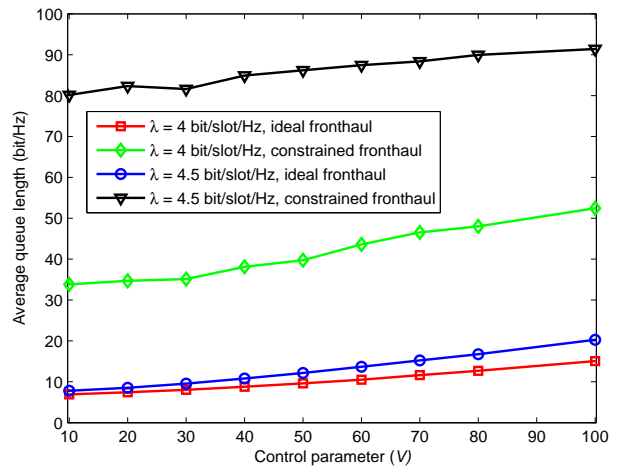}
\vspace*{-10pt} \caption{Average queue length versus
$V$}\vspace*{-10pt} \label{Q_V_Fh}
\end{figure}

The average weighted EE performance under the ideal and constrained
fronthaul are compared in Fig. \ref{RC_V_Fh}, where the average
performance of $\eta_{EE}$ under the fronthaul constraint is better
than that under the ideal fronthaul situation with the same multimedia traffic rate.
Combined with Fig. \ref{Q_V_Fh}, it can be concluded that fronthaul constraint
leads to less power consumption and better energy efficiency performance but
at the cost of incurring larger queueing delay.

\begin{figure}[!h!t]
\centering
\includegraphics[width=0.5\textwidth]{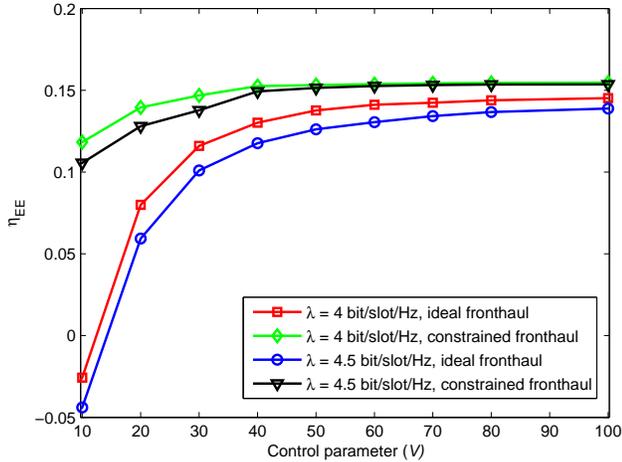}
\vspace*{-10pt} \caption{Average weighted EE performance versus
$V$}\vspace*{-10pt} \label{RC_V_Fh}
\end{figure}

\begin{remark} In a practical H-CRAN, each RUE is associated with only a small number of adjacent RRHs.
Thus, to decrease the computational and operational complexity of the simulation, we first
conduct the simulation in a small area with 2 RRHs and 4 RUEs. In fact,
when the simulated network size increases, similar simulation results are achieved though the simulation time becomes long.
To illustrate this, we conduct another simulation configuration with 8 RRHs and 16 RUEs in the same concerned area. The obtained average queue
lengths for these two simulation configurations are compared in Fig. \ref{Q_V_1}.
It can be observed that the average queue lengths under these two simulation configurations are nearly equal, both grow linearly at $O(V)$.
The average performance of $\eta_{EE}$ is compared in Fig. \ref{EE_V_1}, which grows at the same tendency
with $V$ in Fig. \ref{Q_V_1}. Therefore, we can conclude that the growth rate of the considered performance with the same simulation area are almost stable when the
simulation network size becomes large.
\begin{figure}[!h!t]
\centering
\includegraphics[width=0.5\textwidth]{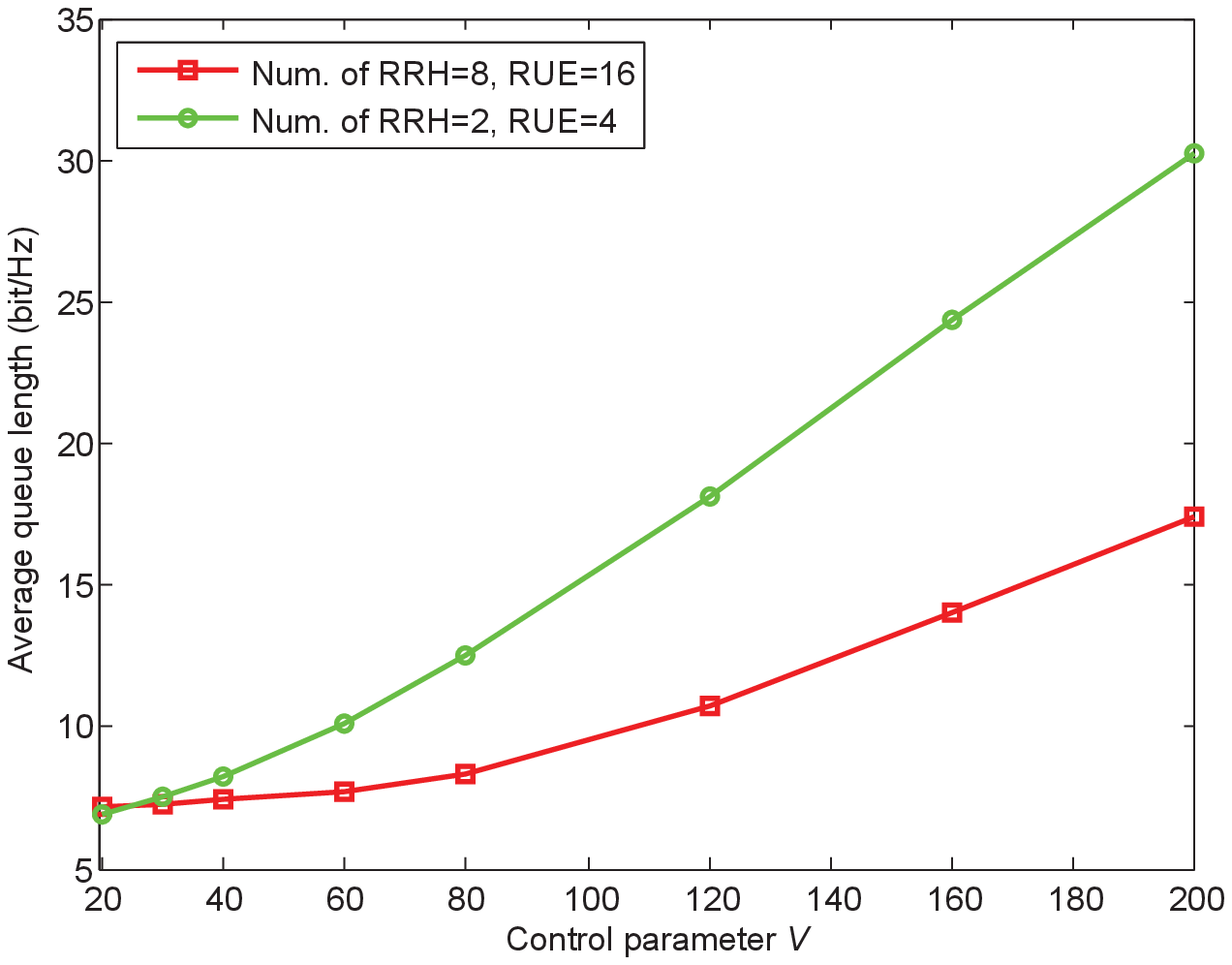}
\vspace*{-10pt} \caption{Average power consumption versus
$V$}\vspace*{-10pt} \label{Q_V_1}
\end{figure}

\begin{figure}[!h!t]
\centering
\includegraphics[width=0.5\textwidth]{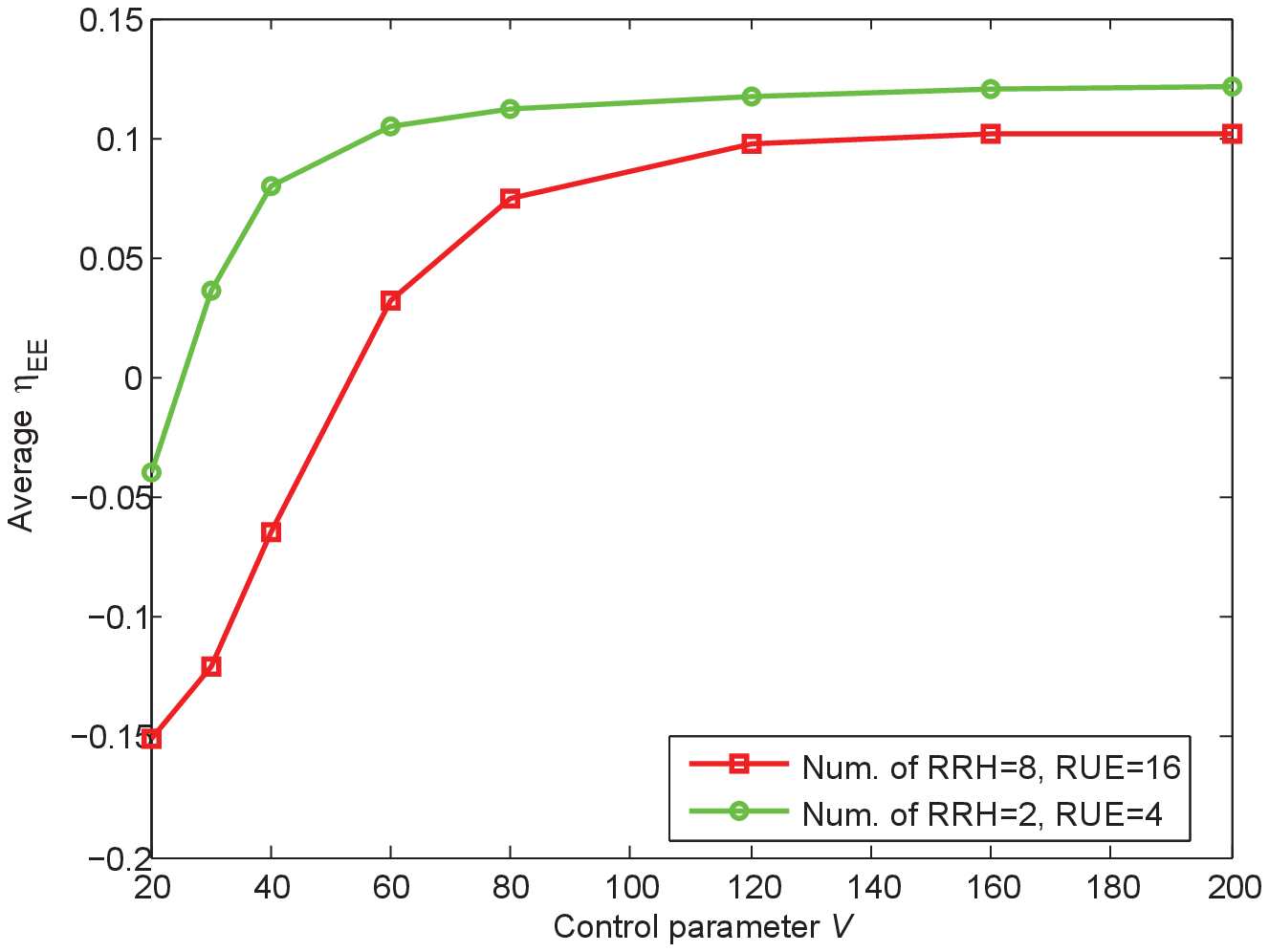}
\vspace*{-10pt} \caption{Average power consumption versus
$V$}\vspace*{-10pt} \label{EE_V_1}
\end{figure}
\end{remark}

\section{Conclusion}

In this paper, to make the average energy efficiency
arbitrarily close to the optimum and make each user's queue
stable in multimedia H-CRANs, an average weighted EE performance metric has
been proposed. Based on the advanced EE performance metric, a
dynamic network-wide beamformer design algorithm based on the
Lyapunov optimization framework has been proposed, which takes the
average and instantaneous power constraints and the interference
constraints into account. This network-wide beamformer design
algorithm can be used to solve the non-convex average weighted EE
performance optimization problem via a general weighted minimum mean
square error (WMMSE) approach. An
$[\mathcal{O}(1/V),\mathcal{O}(V)]$ EE-delay tradeoff is
finally achieved by the proposed algorithm, which is verified by
both the mathematical analysis and numerical evaluations. The
results have shown that the optimal average weighted EE performance
under varied queue lengths strictly depends on the control
parameter $V$. Furthermore, the fronthaul constraint has a
significant impact on the average weighted EE performance. In
realistic multimedia H-CRANs, the optimal $V$ should be pre-selected to
optimize the average weighted EE performance with both ideal and
constrained fronthaul under the given multimedia queuing delay configuration.
\begin{appendices}

\vspace{-10pt}
\section{PROOF OF Lemma 2}

By squaring both sides of (\ref{eq:I3}), the following inequality
can be obtained
\begin{equation}
\label{eq:APP_A1}
\begin{split}
Q_k^2(t + 1)\le &Q_k^2(t) + R_k^2(t) + A_k^2(t)- 2{Q_k}(t){R_k}(t)\\
& + 2{A_k}(t){\left\{ {{Q_k}(t) - {R_k}(t)} \right\}^ + }\\
\le &Q_k^2(t) \!+ \!R_k^2(t) \!+\! A_k^2(t) \!+\!2{Q_k}(t)\left\{ {{A_k}(t)\! -\!{R_k}(t) } \right\}.
\end{split}
\end{equation}

Summing (\ref{eq:APP_A1}) over $k \in \left\{1,2,\ldots,{K_R} \right\}$, we obtain
\begin{equation}
\label{eq:APP_A2}
\begin{split}
&\frac{1}{2}\left\{ {\sum\limits_{k = 1}^{{K_R}} {Q_k^2(t + 1)} \! -
\!\sum\limits_{k = 1}^{{K_R}} {Q_k^2(t)} } \right\} \\
\le&
\frac{1}{2}\sum\limits_{k = 1}^{{K_R}} {\left\{ {R_k^2(t) +
A_k^2(t)}\! \right\}}\!\! -\! \sum\limits_{k = 1}^{{K_R}} {{Q_k}(t)
\left\{ {{R_k}(t) \!-\! {A_k}(t)} \right\}}.
\end{split}
\end{equation}

Similarly, for virtual queues $H_n(t)$, we have
\begin{equation}
\label{eq:APP_A3}
\begin{split}
&\frac{1}{2}\left\{ {\sum\limits_{n = 1}^N {H_n^2(t + 1)} \! -
\!\sum\limits_{n = 1}^N {H_n^2(t)} } \right\} \\
\le&
\frac{1}{2}\sum\limits_{n = 1}^N {{\!\left\{ {{P_n}(t)\! -
\!\!P_n^{avg}(t)} \right\}}^2}\!\!+\!\! \sum\limits_{n = 1}^N
{{H_n}(t)\left\{ {{P_n}(t) \!-\! P_n^{avg}(t)} \!\right\}}.
\end{split}
\end{equation}

Summing (\ref{eq:APP_A2}) and (\ref{eq:APP_A3}) and and taking
expectations of both sides to yield
\begin{equation}
\label{eq:APP_A4}
\begin{split}
&\mathbb{E}\left\{ {L\left( {{\bf{\Theta }}\left( {t + 1} \right)}
\right) - L\left( {{\bf{\Theta }}\left( t \right)} {|{\bf{\Theta
}}\left( t \right)}\right)} \right\} \\
\le& \frac{1}{2}\sum\limits_{k = 1}^{{K_R}} {\mathbb{E}\left\{ {R_k^2(t) + A_k^2(t)} {|{\bf{\Theta }}\left( t \right)} \right\}}\\
&+ \sum\limits_{k = 1}^{{K_R}} {{Q_k}(t)\mathbb{E}\left\{ {
{A_k}(t)-{R_k}(t)} {|{\bf{\Theta }}\left( t \right)}\right\}} \\
&+ \frac{1}{2}\sum\limits_{n = 1}^N {\mathbb{E}{{\left\{ {{P_n}(t) - P_n^{avg}(t)} {|{\bf{\Theta }}\left( t \right)}\right\}}^2}}\\
&+\sum\limits_{n = 1}^N {{H_n}(t)\mathbb{E}\left\{ {{P_n}(t) - P_n^{avg}(t)} {|{\bf{\Theta }}\left( t \right)}\right\}}.
\end{split}
\end{equation}

Subtracting $V\mathbb{E}{\left\{ {{\eta}_{EE}(t)|{\bf{\Theta }}(t)}\right\}}$, we have
\begin{equation}
\label{eq:APP_A5}
\begin{split}
&\Delta \left( {{\bf{\Theta }}\left( t \right)}
\right) - V\mathbb{E }{\left\{ {{\eta}_{EE}(t)|{\bf{\Theta }}(t)}
\right\}} \\
\le &B - V\mathbb{E }{\left\{ {{\eta}_{EE}(t)|{\bf{\Theta }}(t)} \right\}} \\
&+\sum\limits_{k = 1}^{{K_R}} {{Q_k}(t)\mathbb{E}\left\{ {
{A_k}(t)-{R_k}(t) } {|{\bf{\Theta }}\left( t \right)}\right\}} \\
&+ \sum\limits_{n = 1}^N {{H_n}(t)\mathbb{E}\left\{ {{P_n}(t) -
P_n^{avg}(t)} {|{\bf{\Theta }}\left( t \right)}\right\}}.
\end{split}
\end{equation}
where $B$ satisfies (\ref{eq:II13}).

\vspace{-10pt}
\section{PROOF OF THEOREM 1}
Since we use a \emph{C-additive approximation algorithm}, which yields a value within a constant $C$ of the infimum of the R.H.S of (\ref{eq:II12}), it is easy to obtain the following:
\begin{equation}
\label{eq:APP_B1}
\begin{split}
& \Delta \left( {{\bf{\Theta }}(t)} \right) - V\mathbb{E} \left\{
{{\eta}_{EE}(t)|{\bf{\Theta }}(t)} \right\} \\
\le& B + C- V\mathbb{E}\left\{ {{\eta}_{EE}^*(t)|{\bf{\Theta }}(t)} \right\}\\
&+ \sum\limits_{n = 1}^N {{H_n}(t)\mathbb{E}\left\{ {P_n^*(t) -
P_n^{avg}|{\bf{\Theta }}(t)} \right\}} \\
&+ \sum\limits_{k = 1}^{{K_R}}
{{Q_k}(t)\mathbb{E}\left\{ {{A_k}(t) - R_k^*(t)|{\bf{\Theta }}(t)}
\right\}},
\end{split}
\end{equation}
where $R_k^*(t)$, $P_n^*(t)$ and ${\eta}_{EE}^*(t)$ are corresponding values for stationary algorithm referred to in \emph{Lemma 3}. Substituting (\ref{eq:IV24}) into (\ref{eq:APP_B1}) and taking the limit as $\theta \to 0$ leads to:
\begin{equation}
\label{eq:APP_B2}
\Delta \left( {{\bf{\Theta }}(t)} \right) - V\mathbb{E}\left\{ {{\eta}_{EE}(t)|{\bf{\Theta }}(t)} \right\} \le B + C - V{\eta}_{EE}^{opt} - \sum\limits_{k = 1}^{{K_R}} {\epsilon {Q_k}(t)}.
\end{equation}

{\color{red}Summing (\ref{eq:APP_B2}) over $t \in \{0,2,\cdots,T-1 \}$, we obtain
\begin{equation}
\label{eq:APP_B3}
\begin{split}
&\mathbb{E}\left\{ {L({\bf{\Theta }}(T))} \right\} - \mathbb{E}\left\{ {L({\bf{\Theta }}(0))} \right\} - V\sum\limits_{t = 0}^{T-1} {\mathbb{E}\left\{ {{\eta}_{EE}(t)|{\bf{\Theta }}(t)} \right\}} \\
 \le & (B+C)T - VT{{\eta}_{EE}^{opt}}- \sum\limits_{t = 0}^{T-1} {\sum\limits_{k = 1}^{{K_R}} {\epsilon {Q_k}(t)} }.
\end{split}
\end{equation}
}

Based on the fact that $Q_k(t)\geq 0$ for all $t$ and the assumption
(\ref{eq:V2}), we rearrange (\ref{eq:APP_B3}) and obtain
\begin{equation}
\label{eq:APP_B4}
\begin{split}
&\mathbb{E}\{Q_k^2(T)\}\\
\le & 2(B + C - V \eta_{EE}^{opt})
+ 2VT{\eta_{EE}^{max}} +2\mathbb{E}\{L(\boldsymbol{\Theta}(0))\}.
\end{split}
\end{equation}

According to the fact that $\{\mathbb{E}\{|Q_k(T)|\}\}^2 \le \mathbb{E}\{Q_k^2(T)\}$, we have
\begin{equation}
\label{eq:APP_B5}
\begin{split}
&\mathbb{E}\{|Q_k(T)|\} \! \\
\le & \sqrt{2T(B \!+ \!C \!- \!V \eta_{EE}^{opt})\! +\! 2VT{\eta_{EE}^{max}} \!+\!2\mathbb{E}\{L(\boldsymbol{\Theta}(0))\}}.
\end{split}
\end{equation}

Dividing (\ref{eq:APP_B5}) by $T$ and taking limits as $T \to \infty$
\begin{equation}
\label{eq:APP_B6}
\mathop {\lim}\limits_{T \to \infty} \frac{\mathbb{E}\{Q_k(T)\}}{T} = 0.
\end{equation}

Thus, queues are mean-rate stable from \emph{Definition 1}, which indicates that
constraint $C2$ is satisfied according to the proposed algorithm. Similar proof
can be applied to $H_n(t)$.

\vspace{-10pt}
\section{PROOF OF THEOREM 2}
{\color{red}
Based on the inequality (\ref{eq:APP_B3}) obtained in Appendix B and $\mathbb{E}\left\{L({\bf{\Theta }}(0))\right\} < \infty$,
we obtain
\begin{equation}
\label{eq:APP_C1}
\begin{split}
 V\sum\limits_{t = 0}^{T-1} {\mathbb{E}\left\{ {{\eta}_{EE}(t)|{\bf{\Theta }}(t)} \right\}}
\!\ge\! VT{{\eta}_{EE}^{opt}}\!-\!(B+C)T\! - \! \mathbb{E}\left\{ {L({\bf{\Theta }}(0))} \right\} ,
\end{split}
\end{equation}
with some non-negative terms neglected when appropriate.
}

Dividing both sides of (\ref{eq:APP_C1}) by $VT$ and taking the limit as $T \to \infty $, we obtain
 \begin{equation}
 \label{eq:APP_C2}
 \overline {{\eta}_{EE}} \ge {\kern 1pt} {{\eta}_{EE}^{opt}} - \frac{B+C}{V}.
 \end{equation}

Thus \emph{Theorem 2} is proved.

\vspace{-10pt}
\section{PROOF OF THEOREM 3}
{\color{red}
According to the fact that $\mathbb{E}\left\{L({\bf{\Theta
}}(T))\right\} < \infty$, (\ref{eq:APP_B3}) can be re-written as
\begin{equation}
\label{eq:APP_D1}
\begin{split}
\sum\limits_{t = 0}^{T-1} {\sum\limits_{k = 1}^{{K_R}} {\epsilon {Q_k}(t)} } \le& (B+C)T - VT{{\eta}_{EE}^{opt}} - \mathbb{E}\left\{ {L({\bf{\Theta }}(T))} \right\} \\
&+ \mathbb{E}\left\{ {L({\bf{\Theta }}(0))} \right\} +  V\sum\limits_{t = 0}^{T-1} {\mathbb{E}\left\{ {{\eta}_{EE}(t)|{\bf{\Theta }}(t)} \right\}}\\
\le& (B+C)T - VT{{\eta}_{EE}^{opt}}- \mathbb{E}\left\{ {L({\bf{\Theta }}(T))} \right\} \\
&+ \mathbb{E}\left\{ {L({\bf{\Theta }}(0))} \right\} + VT {{\eta}_{EE}^{max}(t)}.
\end{split}
\end{equation}
}

Dividing (\ref{eq:APP_D1}) by $\epsilon T$ and taking the limit as $T \to \infty$, the following is obtained:
\begin{equation}
\label{eq:APP_D2}
\overline Q \le \frac{{B +C + V\left( {{\eta}_{EE}^{max} } -{{\eta}_{EE}^{opt}}\right)}}{\epsilon }.
\end{equation}

This completes the proof of \emph{Theorem 3}.
\end{appendices}

\end{document}